\documentclass[a4paper,reprint,onecolumn,notitlepage,aip,cha,nofootinbib]{revtex4-2}
\usepackage[margin=1in,includeheadfoot]{geometry}
\usepackage[english]{babel}
\usepackage[T1]{fontenc}
\usepackage{tgtermes} 
\usepackage{tgheros} 

\usepackage{amsmath}
\usepackage{amssymb}
\usepackage{amsthm}
\usepackage{mathtools}
\usepackage{bm}
\usepackage{braket}

\usepackage{graphicx}
\usepackage[dvipsnames]{xcolor}
\usepackage{tikz}
\usetikzlibrary{3d}

\usepackage{bbm}
\usepackage{array}
\usepackage{hhline}
\usepackage{mdframed}
\mdfsetup{
  innertopmargin=2pt,
  innerbottommargin=8pt
}

\PassOptionsToPackage{hyphens}{url} 
\usepackage[breaklinks=true]{hyperref}
\bibpunct{[}{]}{;}{n}{}{} 

\usepackage{enumerate}
\usepackage[shortlabels]{enumitem}
\newmdtheoremenv{theorem}{Theorem}
\newtheorem{lemma}[theorem]{Lemma}
\newtheorem{corollary}[theorem]{Corollary}
\newtheorem{proposition}[theorem]{Proposition}

\theoremstyle{remark}

\newtheorem{example}{Example}[section]

\theoremstyle{definition}
\newmdtheoremenv{definition}[theorem]{Definition}

\numberwithin{equation}{section}
\numberwithin{theorem}{section}
\renewcommand{\theequation}{\arabic{section}.\arabic{equation}}
\renewcommand{\thetheorem}{\arabic{section}.\arabic{theorem}}

\newcommand{\rmi}{\mathrm{i}}

\newcommand{\rmd}{\mathrm{d}}

\newcommand{\R}{\mathbb{R}}
\newcommand{\C}{\mathbb{C}}

\renewcommand{\H}{\mathcal{H}}
\renewcommand{\S}{\mathcal{S}}
\newcommand{\V}{\mathcal{V}}
\newcommand{\densmat}{\mathcal{E}}
\newcommand{\puremf}{\mathbb{P}} 

\newcommand{\Bpure}{\mathcal{B}_\mathrm{p}}

\newcommand{\Bens}{\mathcal{B}_\mathrm{e}}

\newcommand{\Bcrit}{\mathcal{B}_\mathrm{c}}
\newcommand{\Breg}{\mathcal{B}_\mathrm{r}}
\newcommand{\Fpure}{\mathcal{F}_\mathrm{p}}

\newcommand{\Fens}{\mathcal{F}_\mathrm{e}}
\newcommand{\FHK}{\mathcal{F}_\mathrm{HK}}
\DeclareMathOperator{\trace}{Tr}

\DeclareMathOperator{\conv}{conv}
\DeclareMathOperator{\rank}{rank}
\DeclareMathOperator{\dom}{dom}
\DeclareMathOperator{\specdown}{spec^\downarrow}
\DeclareMathOperator{\sgrad}{sgrad}
\DeclareMathOperator{\Ad}{Ad}
\DeclareMathOperator{\aff}{aff}
\DeclareMathOperator{\relint}{rel\,int}
\DeclareMathOperator{\relboundary}{\partial_{\rm{rel}}}

\newcommand{\rr}{\bm{r}}
\newcommand{\Fw}{\mathcal{F}_w}
\DeclareMathOperator{\densmap}{\mu}
\newcommand{\up}{\uparrow}
\newcommand{\dn}{\downarrow}
\newcommand\ketbra[1]{\ket{#1}\!\bra{#1}}
\newcommand{\Hpos}{H_+}
\newcommand{\Lsa}{\mathcal{L}_\mathrm{sa}}

\def\CC{\mathbbm{C}}

\begin{document}

\title{Unified Framework for Functional Theories of Quantum Systems}

\author{Chih-Chun Wang}
\affiliation{Arnold Sommerfeld Center for Theoretical Physics, \mbox{Ludwig-Maximilians-Universit\"at M\"unchen, Germany}}

\author{Julia Liebert}
\affiliation{Department of Chemistry, Princeton University, Princeton, New Jersey, USA}
\affiliation{Arnold Sommerfeld Center for Theoretical Physics, \mbox{Ludwig-Maximilians-Universit\"at M\"unchen, Germany}}

\author{Markus Penz}
\email{m.penz@inter.at}
\affiliation{Arnold Sommerfeld Center for Theoretical Physics, \mbox{Ludwig-Maximilians-Universit\"at M\"unchen, Germany}}
\affiliation{Department of Computer Science, Oslo Metropolitan University, Oslo, Norway}

\author{Christian Schilling}
\email{c.schilling@lmu.de}
\affiliation{Arnold Sommerfeld Center for Theoretical Physics, \mbox{Ludwig-Maximilians-Universit\"at M\"unchen, Germany}}

\begin{abstract}
We introduce and study a unified framework for density-functional theory and its variants for quantum systems on finite-dimensional Hilbert spaces. These theories seek to reduce the complexity inherent in the many-body quantum problem by describing ground states through reduced variables. The central ingredients of our unified framework are a generalized choice of basic observables, whose expectation values define precisely those reduced variables, and a fixed part of the Hamiltonian characterizing the class of quantum systems under consideration. It is this minimal structure, which we call the scope of a functional theory, that is necessary and sufficient for the formulation of a functional theory. In particular, it allows one to define the universal functionals, establish their convexity and differentiability properties, address representability questions, and prove a Hohenberg--Kohn-type uniqueness result.
A purification construction also relates ensemble and weighted-ensemble functionals to the pure-state variant. 
Particular emphasis is placed on functional theories with Lie-algebra observable structures, connecting the variational framework to symplectic geometry. The result of this work is a systematic mathematical formulation in which structural results can be proved once and applied across a broad class of finite-dimensional functional theories.
\end{abstract}

\maketitle
\tableofcontents

\section{Introduction and outline}
\label{sec:intro}

The success of density-functional theory (DFT)~\cite{hohenberg-kohn1964,KS1965,Kohn1999-Nobel-Lecture,Burke2012} as an in principle exact reformulation of the quantum many-body ground-state problem has motivated both rigorous mathematical foundations~\cite{Levy79,levy1982,Lieb1983,Garrigue2018HK,Sutter2024,CarvalhoCorso2025} and numerous extensions based on reduced variables other than, or in addition to, the particle density. Examples include current DFT~\cite{Vignale1987,vignale-rasolt-geldart1990,Tellgren2012,Penz2023-CDFT-Review}, spin DFT~\cite{Barth1972,Gunnarsson1976,Vignale1988}, one-particle reduced density matrix functional theory (RDMFT)~\cite{Gilbert1975, Levy79,Valone80, PG16}, as well as functional theories for model Hamiltonians~\cite{maoTestingDensityFunctional2021, Bakkestuen2025-Dicke}. Although these theories differ in their choice of reduced variables and the definition of the universal functional, they share a common variational structure: external fields couple linearly to selected observables, ground-state energies arise by minimization, and universal functionals are obtained by constrained search over quantum states. Many mathematical results are therefore not specific to a particular functional and reduced variable, still they are often proved separately in each theory.

In this work, we develop a unified mathematical framework for ground-state functional theories on finite-dimensional Hilbert spaces. Instead of focusing exclusively on the usual one-particle density variable of DFT, we consider a general space of basic observables and define the reduced variable as the linear map that assigns to each observable its expectation value in a given quantum state.
Different choices of this observable space then recover the usual density, spin-density and current-related variables, one-particle reduced density matrices, and other reduced descriptions. 
A formulation on finite-dimensional Hilbert spaces arises not only after discretization for numerical purposes but is also intrinsic to lattice models, tight-binding and Hubbard-type systems, quantum spin models, and related finite quantum systems. It retains the central structural questions of functional theories while avoiding some of the technical difficulties of the infinite-dimensional Banach space setting~\cite{ENGLISCH1983,Lammert2007,Kryachko2014,teale2022round-table,MY-DFT-Perspective}. Recently growing interest in DFT in finite-dimensional settings~\cite{Wu2006,Pastor2011a,nagyDensityFunctionalTheory2013,nagyQuantumPhaseTransitions2013,Palamara2024,Zhao2025Fractionalized,Penz2025-cs-imag-time,Cances2026-TDDFT-geom-lattice}, which even revealed new mathematical structures~\cite{Xu2022,Song2023,Song2025,penz2023geometry,wangGeometryGeneralizedDensity2026}, makes this elaboration particularly timely.

Similar generalizations for functional theories have been proposed before on various occasions~\cite{Bauer1983,Schonhammer1995,higuchi2004arbitrary,higuchi2004arbitrary2, Liebert2023-refining}, but to date those developments have not been organized within a systematic mathematical framework. 
In filling this gap, the main contributions of our work are:
\begin{itemize}[leftmargin=1em]
    \item We provide a unified framework by identifying the minimal structure necessary to define the reduced variables and the universal functionals, hence the whole functional theory for ground-state properties. 
    We call this structure the `scope' of a functional theory (see Definition~\ref{def:scope}). 
    The derived results show all aspects that are important in DFT and related areas, in particular $N$- and $v$-representability, continuity and differentiability of the universal functionals, and the Hohenberg--Kohn theorem.
    On the other hand, the framework is general enough to encompass practically all existing variants of functional theories 
    for ground states in the finite-dimensional setting, as illustrated in Table~\ref{tab:settings} and exemplified in Sections~\ref{sec:ex-geom} and \ref{sec:ex-ft}.

    \item This unification extends to different variants of the universal functional. In Section~\ref{sec:change-of-scope} we show how the scope can be reduced and demonstrate that the ensemble functional and a functional defined from weighted ensembles ($w$-ensemble functional) are both related to the pure-state functional of a \emph{purified} functional theory. This allows us to deduce properties of these two functionals from statements about the pure-state functional.
    
    \item With a precise notion of a functional theory, we are able to move the topic to the mathematical realm and remove any ambiguities from its formulation.
    The unfolding theory has interesting links to various fields of mathematics, like convex analysis, matrix analysis (in particular joint numerical ranges), Lie algebras, and symplectic geometry.
    Results about functional theories can be proved once and will then hold for any variant within the unified framework. This also offers a systematic guide for constructing new functional theories: the choice of basic observables determines the reduced variable, the representability domain, and the universal functionals.
\end{itemize}

The content of the paper is summarized as follows. After the formal definition of a `scope', which fully defines a functional theory in Section~\ref{sec:def}, we study the geometrical properties of the sets of allowed values for the reduced variables in Section~\ref{sec:geom}. These sets are called `observable ranges' and are nothing but the joint numerical ranges of the basic observables. 
Section~\ref{sec:ex-geom} illustrates the observable ranges with several examples.
Section~\ref{sec:functionals} proceeds with the definition and the properties of the universal functionals as the central objects of the theory that allow one to determine the ground-state energy, highlighting our point of view that the existence of the universal functionals is not due to coincidences or specialized constructions, but rather follows from natural assumptions about the family of quantum systems (defined by the scope) one chooses to study. Their effective domains are precisely the observable ranges from the section before. This section also answers important questions about $v$-representability that are of fundamental importance for functional theories, since they allow to find all values for the reduced variables that actually belong to a valid ground state. Interestingly, in the case of the pure-state constrained-search functional, we find that representability can only be guaranteed if excited states are also included. In addition, this section discusses several variants of the Hohenberg--Kohn theorem that state uniqueness properties of representability, and introduces links between representability and Gâteaux differentiability. Then, in Section~\ref{sec:ex-ft}, we illustrate our framework with some of the most important examples of functional theories for applications in chemistry and physics, including DFT on a lattice, RDMFT, and versions of $N$-qubit models. Section~\ref{sec:change-of-scope} shows that by changing the scope, we either achieve a reduced functional theory or, by extending it, can relate the convex ensemble constrained-search functional and the non-convex pure-state constrained-search functional through a purification technique. We also discuss a similar purification technique for $w$-ensemble functional theories. The last large topic presented in Section~\ref{sec:Lie} is a specialized treatment of functional theories where the observables carry a Lie-algebra structure and for which the symplectic geometry of the state manifold becomes relevant. In this setting, techniques from the theory of moment maps provide powerful tools for the analysis of $N$-representability, resulting in a solution to the one-body $N$-representability problem~\cite{klyachkoQuantumMarginalProblem2006, altunbulakPauliPrincipleRepresentation2008, AK08, klyachkoPauliExclusionPrinciple2009}. This is particularly valuable because the general $N$-representability problem for $p$-body reduced density matrices with $p>1$ is known to be QMA-hard~\cite{liuQuantumComputationalComplexity2007}. We conclude this formulation of a comprehensive framework for functional theories on finite-dimensional Hilbert spaces with an outlook on possible extensions in Section~\ref{sec:outlook}.

\begin{table*}[ht!]
\centering
\renewcommand{\arraystretch}{1.3}
\begin{tabular}{|>{\raggedright\arraybackslash}p{2.9cm}|>{\raggedright\arraybackslash}p{2.7cm}|>{\raggedright\arraybackslash}p{3.6cm}|>{\raggedright\arraybackslash}p{4.0cm}|>{\raggedright\arraybackslash}p{1.92cm}|}
\hline
\textbf{setting} & \textbf{Hilbert space} & \textbf{basic observables} & \textbf{Hamiltonian fixed part (e.g.)} & \textbf{found in Sec.}\\
\hhline{|=|=|=|=|=|}

lattice DFT for $N$ spinless fermions on $M$ sites &
$\H = \wedge^N \C^{M}$ &
site occupations
$n_i=a_i^* a_i$ &
hopping and interaction Hamiltonian $H_0 = \sum_{i\neq j} (t_{ij}a_i^* a_j + w_{ij}a_i^* a_j^* a_i a_j)$ &
\ref{sec:lattice-geom}, \ref{sec:lattice-DFT}, \ref{sec:symplectic:lattice-dft}
\\
\hline

spin system, $N$-qubit model &
$\H=(\C^{2})^{\otimes N}$ &
local spin operators $X_i, Y_i, Z_i$ (full Bloch vector or reduced set) &
spin chain 
$H_0 = -\sum_{i=1}^{N-1} X_iX_{i+1}$ &
\ref{sec:Bloch}, \ref{sec:geometry_double_qubit}, \ref{sec:spin-chain}, \ref{sec:double-qubit-FT}
\\
\hline

momentum functional theory for bosons on periodic $L^d$-site lattice &
subspace $\mathcal{H}_{\bm{p}}$ with fixed total momentum $\bm{p}\in (\mathbb{Z}/L\mathbb{Z})^d$ &
momentum occupations
$n_{\bm{k}}$&
on-site interaction $H_0=\sum_i b_i^*b_i(b_i^*b_i-1)$ &
\ref{sec:momentum-ft}
\\
\hline

lattice RDMFT for $N$ spinless fermions on $M$ sites &
$\H = \wedge^N \C^{M}$ &
self-adjoint one-body operators spanned by $a_i^*a_j$
&
interaction Hamiltonian
$W= \sum_{i\neq j} w_{ij}a_i^* a_j^* a_i a_j$
&
\ref{sec:RDMFT}, \ref{sec:LieRDMFT}
\\
\hline
spin-independent RDMFT &
$\H = \wedge^N (\C^{M}\otimes\C^2)$ &
generators of $\rm U(d)$ $E_{ij} = a_{i\uparrow}^* a_{j\uparrow}+ a_{i\downarrow}^* a_{j\downarrow}$ &
spin-independent ($\mathrm{SU}(2)$-invariant) interaction Hamiltonian $W$
& \ref{sec:symplectic:spin}
\\
\hline
general spin-adapted RDMFT &
$\H = \wedge^N (\C^{M}\otimes\C^2)$ &
generators of $\rm U(d)$ $E_{ij}$ and generators of $\rm U(2)$ $S_{\sigma\sigma'}=\sum_{i=1}^Ma_{i\sigma}^* a_{i\sigma'}$&
interaction Hamiltonian $W$ 
& \ref{sec:symplectic:spin}
\\
\hline

spin-resolved lattice DFT (Hubbard model) &
$\H = \wedge^N (\C^{M}\otimes\C^2)$ &
spin-resolved site occupations
$n_{i\sigma}=a_{i\sigma}^* a_{i\sigma}$ where $\sigma\in\{\uparrow,\downarrow\}$&
Hubbard Hamiltonian
$H_0 = -\sum_{i\neq j, \sigma\in\{\uparrow,\downarrow\}} t_{ij}a_{i\sigma}^* a_{j\sigma} + U\sum_i a_{i\uparrow}^* a_{i\downarrow}^* a_{i\uparrow} a_{i\downarrow}$
& -- \\
\hline
lattice current DFT &
$\H = \wedge^N \C^{M}$ &
site occupations $n_i=a_i^* a_i$ and link currents
$J_{ij}=-\rmi\!\left(t_{ij}a_i^* a_j-t_{ji}a_j^* a_i
\right)$
&
same as in lattice DFT (but with $t_{ij}=\overline{t_{ji}}\in\C$)
& --
\\
\hline
\end{tabular}
\caption{Examples of functional theory settings within the unified framework. The basic observables always define the reduced variable that enters the functionals by expectation values.}
\label{tab:settings}
\end{table*}

\section{Scope of a functional theory}
\label{sec:def}

In this section, we specify the fundamental mathematical structures that
constitute a functional theory's `scope', from which all other objects of the
theory will be derived.

Functional theories aim to reformulate the ground-state problem in terms of reduced variables rather than the full quantum state. 
The reduced variables are given as the expectation values of observables from a chosen subspace of $\Lsa(\H)$, the space of all self-adjoint operators on $\H$. 
At the same time, 
this subspace plus a fixed $H_0\in\Lsa(\H)$ defines the set of all possible 
Hamiltonians of the theory, 
allowing one to capture whole physical fields adequately in one single setting. For example, in quantum chemistry, $H_0$ will always include the Coulomb repulsion and kinetic energy of electrons, and only the external potential from the nuclei is allowed to vary. We now formalize this as follows:

\begin{definition}\label{def:scope}%
    The defining \emph{scope} of a functional theory is a tuple $\S=(\V, \H, \iota, H_0)$, where $\V$ is a real vector space, $\H$ a finite-dimensional Hilbert space, $\iota: \V\rightarrow \Lsa(\H)$ a linear map, and $H_0\in\Lsa(\H)$ a self-adjoint operator on $\H$.\\
    If $\iota(\V)$ is such that all its elements commute pairwise, then it is called an \emph{abelian scope}.
\end{definition}

Various aspects of the structure encoded in Definition~\ref{def:scope} have appeared in previous work on functional theories~\cite{Liebert2023-refining, LMS24-jctc, wangGeometryGeneralizedDensity2026, Liebert25-phdthesis}. The present formulation identifies this structure in a concise and rigorous manner and uses it as the foundation for a unified mathematical framework. In particular, Definition~\ref{def:scope} determines the admissible Hamiltonians and reduced variables, and leads to a canonical definition of the eponymous universal functionals introduced in Section~\ref{sec:functionals}.

We give the separate notion of an `abelian scope', since some of the following results will only hold for the abelian case.
The operator $H_0$ from the definition is seen as the fixed, universal part of a Hamiltonian describing a quantum system, i.e., the part of the system that we cannot manipulate. A scope then yields an affine subspace of  $\Lsa(\H)$ that describes the set of all Hamiltonians under consideration,
\begin{equation}\label{eq:ham}
    H(v) := H_0 + \iota(v),
\end{equation}
where $v\in\V$ parameterizes the variable part of the Hamiltonian, called the \textit{(external) potential}. 
For example, consider a system consisting of $N$ identical particles on $M$ lattice sites (see Section~\ref{sec:lattice-geom}).
We can choose $\V=\mathbb{R}^M$ and $\iota: \R^M\rightarrow \Lsa(\H), v\mapsto \sum_{i=1}^M v_ia_i^* a_i$ to parameterize the space of all local potentials. This already demonstrates one reason behind this kind of parametrization, since the same $\V$ can be mapped to potential operators acting on any number (or kind) of particles.

Note that if $\iota$ is not injective or $I_{\H} \in \iota(\V)$, then the
parameterization of the Hamiltonians carries redundancy. It can be desirable to
choose a scope that avoids such redundancy, especially with respect to results
relating to `critical values' (see Section~\ref{sec:crit} and note its relevance
for the strong Hohenberg--Kohn result in Theorem~\ref{th:strongHK}).

In practice, we can always choose a basis $(e_1,\ldots,e_m)$ of $\V$ that is
mapped to $B_i=\iota(e_i)\in\Lsa(\H)$. For every $v\in\V$ we can then write
\begin{equation}
    \iota(v)=\sum_{i=1}^m v_i\iota(e_i)=\sum_{i=1}^m v_iB_i.
\end{equation}
The $B_i$ can be physically interpreted as \emph{control} and \emph{measurement} observables. The first interpretation is clear from the form of the Hamiltonian in Eq.~\eqref{eq:ham}.
On the other hand, 
their connection to measurements is through the definition of a reduced variable as the collection of their expectation values, as we elaborate below. 
This reduced variable then generalizes the notion of particle density in ordinary DFT, as seen in the example discussed in Section~\ref{sec:lattice-geom}.

Let $\Psi \in \mathcal{H}$ be a normalized vector, and let $P_\Psi$ be the corresponding  rank-one orthogonal projector (also written as $\ket{\Psi}\!\bra{\Psi}$). We define its \textit{density} with respect to the scope $\S$ as
the linear functional on $\V$ given by
\begin{equation}\label{eq:density-map}
    \densmap(P_\Psi) := (v \mapsto \braket{\Psi, \iota(v)\Psi} 
     =\trace(P_\Psi\iota(v)))
    \in \V^*.
\end{equation}
Here and in the following, we adopt the convention that the inner product $\braket{\cdot,\cdot}$ on $\mathcal{H}$ is antilinear in the first entry and we sometimes use the notation $\langle A\rangle_\Psi=\langle\Psi,A\Psi\rangle$ for the expectation value.
Hence, for any quantum state $P_\Psi$, the associated density $\rho = \mu(P_\Psi)$ belongs to the dual $\V^*$, which can be identified with $\V$ after choosing an inner product on $\V$.
As such, a density takes a potential as an element of $\V$ and maps it to a scalar value that represents the potential energy within that potential.
In the case of DFT (Section~\ref{sec:lattice-DFT}), the elements of $\V^*$ would then be the usual one-particle densities.
With respect to a chosen basis $(e_i)_{i=1}^m$ for $\V$, 
the components of $\densmap(P_\Psi)$ are given by the tuple of expectation values $(\braket{\Psi, B_1\Psi}, \ldots, \braket{\Psi, B_m\Psi})\in \mathbb{R}^m$ with $B_i=\iota(e_i)$ as above.

The principal aim of a functional theory is now to describe the system in terms of the densities $\rho = \mu(P_\Psi)\in \V^*$ 
instead of the complete quantum state $\Psi$. 
This is possible in the variational problem for the ground-state energy since we can pass from expectation values of the Hamiltonian to constrained-search functionals on $\V^*$, defined in Section~\ref{sec:functionals}, for the class of physical situations covered by the scope.
In choosing the observable space $\iota(\V)$, a compromise between the necessary flexibility (degree of variety of physical situations) and the complexity of the ensuing theory (which can be measured by $m=\dim\V$) must be made.
Having $m\ll\dim\H$ leads to a considerable reduction in complexity compared to, e.g., a many-particle Hilbert space.
On the other hand, as a trade-off, one loses the linear structure of quantum mechanics, 
which is replaced by 
a non-linear but lower-dimensional theory for obtaining ground-state properties of quantum systems further discussed in Section~\ref{sec:functionals}.
The expectation value of $H(v)$, which yields the total energy of the system and allows to define ground states, can be rewritten with the density map as
\begin{equation}\label{eq:H-mu}
    \langle\Psi, H(v)\Psi\rangle = \langle\Psi,H_0\Psi\rangle +  \langle \mu(P_\Psi),v \rangle.
\end{equation}
Here, the first bracket on the right-hand side denotes the inner product in $\H$ while the second is the dual pairing on $\V^*\times\V$.
Minimizing over all normalized states defines the (ground-state) energy functional $E(v)$ (see Eq.~\eqref{eq:E-def}). 

Finally, we define the density of a general ensemble state $\Gamma \in \Lsa(\mathcal{H})$ with $\trace\Gamma=1$ and $\Gamma\geq 0$
by affine extension of Eq.~\eqref{eq:density-map},
\begin{align}\label{eq:mu-DM}
    &\densmap(\Gamma) := (v\mapsto\trace(\Gamma v)) \in \V^*,\\
    &\trace(H(v)\Gamma) = \trace(H_0\Gamma) + \langle \mu(\Gamma),v \rangle, \label{eq:Hv-exp}
\end{align}
which includes pure states by setting $\Gamma=P_\Psi$.
The sets of all pure states and density matrices will be introduced together with the set of all attainable densities $\densmap(P_\Psi)$ and $\densmap(\Gamma)$, which form the basic domains of the theory in the next section.

\section{Geometry of quantum states and observable ranges}
\label{sec:geom}

In this section, we study the properties of the images of the sets of pure and
ensemble states under the density map $\mu$. 
As will become clear in Section~\ref{sec:functionals},
these sets are of central importance since they are exactly the domains of the universal functionals. Furthermore, the notion of criticality discussed in 
Section~\ref{sec:crit} facilitates the rigorous formulation of the 
strong Hohenberg-Kohn theorem and a regularity result presented in Section~\ref{sec:hk_theorems}.
Since our analysis will depend on
methods from differential geometry, we start by introducing several geometric
notions associated with quantum states.

\subsection{Quantum states}
Given a finite-dimensional Hilbert space $\H$, 
the set of all density operators is the compact and convex set
\begin{equation}\label{eq:densmat}
    \densmat(\H) :=
    \{\Gamma\in \Lsa(\mathcal{H}) : \trace\Gamma=1, \Gamma\geq 0 \}.
\end{equation}
The set of pure states (rank-one density matrices) is identified with the projective Hilbert space and is given by
\begin{equation}\label{eq:pure-states}\begin{aligned}
    \mathbb{P}(\mathcal{H})&:= \{P\in\densmat(\H) : \rank P= 1\} = \{ P_\Psi : \Psi\in\H, \|\Psi\|=1\} \subseteq \densmat(\H).
\end{aligned}\end{equation}
While $\mathbb{P}(\mathcal{H})$ is a manifold, $\densmat(\H)$ is not (unless $\dim \H \le 1$).
Intuitively, this is due to the different ranks of its elements and the existence of a boundary due to positive-semidefiniteness.
Note that both $\mathbb{P}(\mathcal{H})$ and $\densmat(\H)$ are compact as subsets of $\Lsa(\mathcal{H})$.

Since the natural $\mathrm{U}(\mathcal{H})$-action on $\mathbb{P}(\mathcal{H})$ is transitive, at any pure state $P\in \mathbb{P}(\mathcal{H})$, the tangent space $T_P\mathbb{P}(\mathcal{H})$ is spanned by vectors of the form 
\begin{equation}
\frac{\rmd}{\rmd t} e^{tS} P e^{-tS}\Big|_{t=0}=[S,P],
\end{equation}
where $S\in \rmi \Lsa(\H)$.
That is,
\begin{equation}\label{eq:tangent-space}
    T_P\mathbb{P}(\mathcal{H})= \{[S,P] : S\in \rmi \Lsa(\mathcal{H}) \} \subseteq \Lsa.
\end{equation}
This characterization of the tangent space will be useful in Section~\ref{sec:Lie}, when we discuss a symplectic-geometric formulation of functional theories. 

\subsection{Observable ranges and \texorpdfstring{$N$}{N}-representability}

For the rest of this section, let $\S =(\V,\H, \iota,H_0)$ be a scope according to Definition~\ref{def:scope} and let $\mu: \densmat(\H)\rightarrow \V^*$ denote the corresponding density map.

The basic geometric objects in our mathematical investigation of generalized functional theories will be the effective domains of the constrained-search functionals defined in Section~\ref{sec:functionals}. 
As will become clear in Section~\ref{sec:functionals}, these sets are given by the images of the sets $\puremf(\mathcal{H})$ and $\densmat(\H)$ under the density map $\densmap$. 

\begin{definition}
    We define the \emph{pure-state} and \emph{ensemble observable ranges} as the images of $\puremf(\H)$ and $\densmat(\H)$ under $\densmap$:
    \begin{equation}
        \Bpure := \densmap(\puremf(\H)) 
        \subseteq \V^*,\qquad \Bens := \densmap(\densmat(\H)) 
        \subseteq \V^*.
    \end{equation}
\end{definition}

We would like to remark that these objects bear different names 
in different fields.
In DFT (in finite dimensions), they are the set of $N$-representable densities~\cite{penz2021-Graph-DFT}, and they coincide since all observables commute (see Proposition~\ref{prop:ranges}\ref{prop:ranges:item:commute} below). In the context of RDMFT, characterizing the set $\Bpure$ is related to the quantum marginal problem~\cite{Schilling2014} and its full characterization involves the so-called generalized Pauli constraints (see also Example~\ref{example:fermion_rdmft}). 
In matrix analysis, the sets are known as the ``joint numerical ranges''~\cite{Muller2020,Plaumann2021}, the study of which can be traced back to a question on the convexity of the numerical range of a single complex matrix put forward by \citet{toeplitz1918} in 1918, which led to the celebrated Toeplitz--Hausdorff theorem~\cite{Davis1971}. As of today, finding conditions for the convexity of $\Bpure$, which is in general a difficult problem, constitutes an entire mathematical subfield. Here, we will only note a few interesting results. 
\citet{Au-Yeung1979} show that for the specific case $\dim\iota(\V)=3$ and $\dim\H\geq 3$ the set $\Bpure$ is always convex and also discuss the situation for general $\dim\iota(\V)$ and for ensemble states. In Section~\ref{sec:Bloch}, we discuss the classical example of the Bloch sphere with $\dim\iota(\V)=3$ and $\dim\H=2$, where $\Bpure$ is not convex. \citet{Gutkin2004} connect the convexity of $\Bpure$ to a property of constant multiplicity of eigenvalues and discusses the convex hull of the joint numerical range that equals $\Bens$ by Proposition~\ref{prop:ranges}\ref{prop:ranges:item:ch}.
The work of \citet{Li2011} link the topic to quantum error correction and provide an infinite-dimensional extension.
\citet{Plaumann2021} deal with properties of $\Bens$ and treat it by means of algebraic geometry, while \citet{Jarov2023} are also concerned with a characterization of the set $\Bens$ and are able to connect it in a one-to-one fashion to expectation values of thermal states.

Since the sets $\Bpure$ and $\Bens$ are the effective domains of the functionals later introduced in Section~\ref{sec:functionals}, it is important to understand their geometric properties, which will be useful for optimization problems.

\begin{proposition}\label{prop:ranges}
    The following properties hold for the observable ranges.
    \begin{enumerate}[(i)]
        \item\label{prop:ranges:item:ch} $\Bens = \conv(\Bpure)$, where $\conv(\cdot)$ denotes the convex hull. In particular, $\Bpure\subseteq\Bens$ and $\Bens$ is convex.
        \item\label{prop:ranges:item:compact} $\Bpure,\Bens \subsetneq \V^*$ are compact.
        \item\label{prop:ranges:item:commute} If the scope $\S$ is abelian, then $\Bpure=\Bens$ and the set is a convex polytope.
        \item\label{prop:ranges:item:aff} 
        If $\iota: \V\rightarrow \Lsa(\H)$ is injective and $I_{\H}\notin \iota(\V)$, then
        $\aff(\Bpure) = \aff(\Bens) = \V^*$,
        where $\mathrm{aff}(\cdot)$ denotes the affine hull. 
    \end{enumerate}
\end{proposition}

\begin{proof}\ref{prop:ranges:item:ch} The set $\Bens$ is convex since it is the image of a convex set under a linear map. 
    Clearly, $\Bens \supseteq \conv \Bpure$ because $\Bens \supseteq \Bpure$ and $\Bens$ is convex. For any $\rho \in \Bens$, there is an ensemble state $\Gamma = \sum_i t_i P_i$ (finite sum) with $P_i\in\puremf(\H)$, $t_i\in[0,1]$, and $\sum_i t_i=1$ such that $\densmap(\Gamma) = \rho$. Since $\densmap$ is linear, we have
    \begin{equation}
       \rho = \sum_i t_i \densmap(P_i) \in \conv\Bpure.
    \end{equation}
    This shows $\Bens \subseteq \conv\Bpure$.\\
    \ref{prop:ranges:item:compact} The sets $\puremf(\H)$ and $\densmat(\H)$ are compact, so their images under the continuous map $\densmap$ are also compact.\\
    \ref{prop:ranges:item:commute} 
    Since all $\iota(v)$, $v\in\V$, commute, we can find an orthonormal basis $(\Phi_k)_{k=1}^{d}$ of $\H\cong\C^d$ such that there exist $\omega_1, \ldots, \omega_{d} \in \V^*$ satisfying
    \begin{equation}
    \forall v\in \V: \iota(v) \Phi_k = \braket{\omega_k, v}\Phi_k.
    \end{equation}
   Any $\Psi\in\H$ can be written as $\Psi=\sum_{k=1}^{d} c_k\Phi_k$, and so for all $v\in \V$
    \begin{equation}
    \begin{aligned}
    &\braket{\densmap(P_\Psi), v}
    = \braket{\Psi,\iota(v)\Psi} = 
    \sum_{k,\ell=1}^{d} \overline c_k c_{\ell}\braket{\Phi_k,\iota(v)\Phi_\ell}
    = \sum_{k=1}^{d} |c_k|^2 \braket{\omega_k, v}.
    \end{aligned}
    \end{equation}
    Since this holds for arbitrary $v\in \V$, we have
    $\densmap(P_\Psi) = \sum_{k=1}^d |c_k|^2 \omega_k \in \conv(\omega_1, \ldots, \omega_{d})$. This shows $\Bpure \subseteq \conv(\omega_1, \ldots, \omega_{d})$. Conversely, take any $\rho \in \conv(\omega_1, \ldots, \omega_{d})$, which can be written as a convex combination $\rho = \sum_{k=1}^{d} t_k \omega_k$. 
    Let $\Psi := \sum_{k=1}^d \sqrt{t_k}\Phi_k$, which is normalized. 
    Then 
    \begin{equation}
    \densmap\left(P_\Psi\right) = \sum_{k=1}^d t_k \omega_k = \rho,
    \end{equation}
    which shows $\conv(\omega_1, \ldots, \omega_{d})\subseteq \Bpure$.
    Hence, $\Bpure$ is a convex polytope. Since $\Bens = \conv\Bpure$, the two sets must coincide.\\
    \ref{prop:ranges:item:aff} 
    For the proof, extend $\mu$ linearly to all of $\Lsa(\H)$.
    First, note that $\mathrm{aff}(\Bpure)$ and $\mathrm{aff}(\Bens)$
    agree because $\Bens = \mathrm{conv}(\Bpure)$. 
    Since the density map $\mu$ is linear, we have $\mathrm{aff}(\Bens) = \mu(\mathrm{aff}(\densmat(\H)) =  \mu(\Gamma_0 + I_\H^\perp) = \mu(\Gamma_0) + \mu(I_\H^\perp)$, where $\Gamma_0$ is any ensemble state, and $I_\H^\perp$ denotes the orthogonal complement of $\mathbb{R}\cdot I_\H\subseteq \Lsa(\H)$
    with respect to the Hilbert--Schmidt inner product $\langle A,B\rangle_{\rm HS}=\trace(A^*B)$.
    Thus, it suffices to show that $\mu(I_\H^\perp) = \V^*$. 
    Since $I_\H\notin \iota(\V)$, we can find a (not necessarily orthogonal) direct sum decomposition $\Lsa(\H) = \iota(\V) \oplus (\R\cdot I_\H) \oplus R$.
    Let $\rho\in\V^*$ be any vector. 
    Because $\iota$ is injective, the dual $\iota^*: \Lsa(\H)^*\rightarrow\V^*$ is surjective, so we can find a $\tilde\rho\in \Lsa(\H)^*$ such that $\iota^*(\tilde \rho) = \rho$ and 
    $\tilde\rho (I_\H) = 0$. Let $\Gamma\in \Lsa(\H)$ be such that $\tilde \rho = \trace(\Gamma \cdot)$. Then $\mu(\Gamma) = \rho$, showing that $\mu$ is surjective from $I_\H^\perp$ onto $\V^*$.
 \end{proof}

\begin{example}
Take $\H = \mathbb{C}^2$. Consider the following two scopes
\begin{equation}
\begin{aligned}
    &\S_1 = (\mathbb{R}^3, \H, v\mapsto v_1X+v_2Y + v_3Z, H_0)\\
    &\S_2 = (\mathbb{R}^4, \H, v\mapsto v_1X+v_2Y + v_3Z + v_4I_{\H}, H_0),
\end{aligned}
\end{equation}
where $X,Y,Z$ are the Pauli matrices and $H_0\in\Lsa(\C^2)$ is arbitrary (irrelevant for the geometry of $\Bpure$ and $\Bens$).
Then $\Bpure^{(\S_1)}$ is the Bloch sphere (unit sphere) in $\mathbb{R}^3$ and 
$\Bpure^{(\S_2)} = \{(x,y,z, 1) \in \mathbb{R}^4: (x,y,z)\in \Bpure^{(\S_1)}\}$.
It follows that both $\aff(\Bpure^{(\S_1)})$ and $\aff(\Bpure^{(\S_2)})$ 
are isomorphic to $\mathbb{R}^3$, but the latter has codimension 1 in the density space $\V^* = \mathbb{R}^4$.
This is caused by the presence of $I_{\H}$ in the
second scope.
 \end{example}

\subsection{Critical and regular values}
\label{sec:crit}

In this section, we divide the set $\Bpure$ into regular and critical values. This will play an important role in uniqueness results (Hohenberg--Kohn theorem) when mapping from densities back to potentials.

\begin{definition}\label{def:crit}
    A pure state $P=P_\Psi\in \puremf(\H)$ is called a \textit{critical point}
    if $\Psi$ is an eigenvector of some $\iota(v)$ satisfying
    $\iota(v)\not\propto I_\H$. 
    A density $\rho\in\V^*$ is a \textit{critical value} if 
    there exists a pure state $P\in \densmap^{-1}(\rho)$ such that $P$ is a critical point.
    The set of all critical values is denoted $\Bcrit\subseteq\Bpure$. The elements in the set $\Breg:=\Bpure\setminus\Bcrit$ are called \emph{regular values}.
\end{definition}

We note that, relative to $\Bpure$, the set $\Bcrit$ is closed and consequently $\Breg$ is open.
We now give an equivalent characterizations of critical points.

\begin{lemma}\label{lem:crit-commute}
A pure state $P\in \puremf(\H)$ is critical if and only if there is a $v \in \V$ with $\iota(v) \not\propto I_\H$ such that $[P,\iota(v)]=0$.
\end{lemma}

\begin{proof}
    If $P=P_\Psi$ is critical, then by definition we have $\iota(v)\Psi=\lambda\Psi$. 
    It follows that $[P,\iota(v)]=\ketbra{\Psi}\iota(v)-\iota(v)\ketbra{\Psi} = \lambda \ketbra{\Psi} - \lambda\ketbra{\Psi} = 0$. \\
    Conversely, note that $[P, \iota(v)]=0$ means that $P$ and $\iota(v)$ share a common eigenbasis. Furthermore, since $P=P_\Psi$ is a one-dimensional projector, $\Psi$ must be one of the basis vectors and also an eigenvector of $\iota(v)$. 
\end{proof}

If the redundancy in the parameterization of the Hamiltonian mentioned after
Definition~\ref{def:scope} is avoided, we have the following further equivalent
characterizations for critical points.

\begin{proposition}\label{prop:crit-equivalent}
    Assume $\iota: \V\rightarrow \Lsa(\H)$ is injective and $I_\H\notin\iota(\V)$. For a basis $(e_i)_{i=1}^m$ of $\V$ set $B_i=\iota(e_i)$. Then the following statements are equivalent.
    \begin{enumerate}[(i)]
        \item\label{item:def} $P_\Psi\in \puremf(\H)$ is a critical point.
        \item\label{item:lin-dep} 
        The vectors $\Psi,B_1\Psi,\ldots,B_m\Psi \in \H$ are linearly dependent over $\R$.
        \item\label{item:Gram} 
        The Gram matrix $G_{ij}=\langle B_i\Psi , B_j\Psi \rangle$, with $i,j\in\{0,1,\ldots,m\}$ and $B_0:=I_\H$, has zero determinant.
        \item \label{item:derivative} The derivative $\rmd_P\densmap: T_P \puremf(\H)\rightarrow \V^*$ is not surjective.
    \end{enumerate}
\end{proposition}

\begin{proof}
\ref{item:def} $\Rightarrow$
\ref{item:lin-dep}
By definition, 
$\iota(v)\Psi=\lambda\Psi$ for a $v\in\V$ with $\iota(v)\not\propto I_\H$, 
which immediately implies that the vectors $\Psi,B_1\Psi,\ldots,B_m\Psi$ in
$\H$ are linearly dependent over $\R$.\\
\ref{item:lin-dep} $\Rightarrow$ \ref{item:def} 
Suppose $\sum_{i=1}^mv_i B_i \Psi = \lambda \Psi$ with
$v_1, \dotsb, v_m, \lambda$ not all zero.
I.e., $\iota(v) \Psi = \lambda \Psi$ 
with $v = \sum_{i=1}^m v_i B_i$.
If $\iota(v)$ were proportional to $I_\H$, then $v=0$ because $\iota$ is injective and $I_{\H}\notin \iota(\V)$, a contradiction.
Hence, we conclude that $\iota(v)\not\propto I_\H$.\\
\ref{item:lin-dep} $\Leftrightarrow$ \ref{item:Gram} Linear dependence of the columns (or rows) of a square matrix is equivalent to the matrix having zero determinant.\\
\ref{item:def} $\Leftrightarrow$ \ref{item:derivative}
For any self-adjoint $A$, we have
\begin{equation}
\left\langle v, \rmd_P\densmap (\rmi[A,P]) \right\rangle
= \trace(\rmi[A,P] \iota(v)).
\end{equation}
The derivative $\rmd_P\densmap$ is not surjective if and only if there is a non-zero $v\in\V$ such that for all self-adjoint $A$ one has $\trace(\rmi[A,P]\iota(v))=0$. 
By the cyclicity of trace, this is equivalent to $[P, \iota(v)]=0$.
By the injectivity of $\iota$, $v\neq 0$ is equivalent to $\iota(v)\neq 0$. 
Finally, $\iota(v)\not\propto I_\H$ by the assumption $I_\H\notin\iota(\V)$.
\end{proof}
It is easy to see that if the `no-redundancy' assumption 
(that $I_{\H}\notin \iota(\V)$ and $\iota$ is injective) is dropped, 
the equivalence between \ref{item:def} and \ref{item:derivative} still holds if we consider surjectivity onto $T_{\mu(P)} \mathrm{aff}(\Bpure)$
for the latter.

We note that the Gram matrix from \ref{item:Gram} was previously used in the
context of uniqueness questions in lattice
DFT~\cite{Xu2022,Penz2025-cs-imag-time} and it is for this reason that we
include it in this list of equivalences. The same form of criticality was
studied by \citet{Song2025} in a very similar setting. Our notion of criticality
coincides with the usual one for smooth maps of smooth manifolds as
\ref{item:derivative} shows, so the set of critical values has measure zero in $\aff(\Bpure)$ by
Sard's theorem \cite{leeIntroductionSmoothManifolds2012}.

We write $\relboundary\mathcal{B}:= \overline{\mathcal{B}} \setminus \relint
\mathcal{B}$ for the relative boundary of a set $\mathcal{B}$. The relative
boundary $\relboundary\Bpure$ of the pure-state observable range directly
relates to critical values, as the following proposition shows.

\begin{proposition}\label{prop:boundary-minimized}
    Let $\rho_0 \in \Bpure$.
    \begin{enumerate}[(i)]
    \item \label{item:prop_boundary_crit_1} If there exists
    a $v \in\V$ with $\iota(v) \not\propto I_\H$ such that 
    for all $P_0 \in (\densmap|_{\puremf(\H)})^{-1}(\rho_0)$, the map $P \mapsto \braket{\densmap(P), v}$ is locally minimized at $P_0$, then $\rho_0 \in \relboundary\Bpure$.

    \item \label{item:prop_boundary_crit_2} If $\rho_0 \in \relboundary\Bpure$, then every $P_0\in (\densmap|_{\puremf(\H)})^{-1}(\rho_0)$ is a critical point. In particular, $\relboundary\Bpure \subseteq \Bcrit$.
    \end{enumerate}
\end{proposition}

\begin{proof}
    \noindent
    \ref{item:prop_boundary_crit_1}
    We will prove the contrapositive. Suppose $\rho_0\in \relint\Bpure$. We will show that for all $v\in  \V$ such that $\iota(v)\not\propto I_{\H}$, there exists a pure state $P_0 \in \densmap^{-1}(\rho_0)$ such that for all open neighborhoods $U\ni P_0$, there exists a $P\in U$ such that $\braket{\densmap(P), v} < \braket{\densmap(P_0), v}$.
    Take an arbitrary $v\in \V$ satisfying $\iota(v)\not\propto I_{\H}$. Since $\rho_0 \in \relint\Bpure$, we can find a sequence $(\rho_n)_{n=1}^\infty$ such that $\rho_n\rightarrow \rho_0$ and $\braket{\rho_n, v} < \braket{\rho_0, v}$ for all $n$. For each $n$, choose a pure state $P_n$ in the preimage $(\densmap|_{\puremf(\H)})^{-1}(\rho_n)$. By compactness, there is a subsequence $(P_{n_k})_{k=1}^\infty$ such that $P_{n_k}\to P_0$ for some $P_0 \in \puremf(\H)$ as $k\to\infty$. By the continuity of $\densmap$, we have $\densmap(P_0) = \rho_0$. If $U\ni P_0$ is any open neighborhood of $P_0$, then $P_{n_k} \in U$ for some $k$. But then $\braket{\densmap(P_{n_k}), v} = \braket{\rho_{n_k}, v} < \braket{\rho_0, v} = \braket{\densmap(P_0), v}$.\\    
   \ref{item:prop_boundary_crit_2}
   Suppose $\rmd_{P_0} (\densmap|_{\puremf(\H)}): T_{P_0}\puremf(\H) \rightarrow \V^*$ is surjective onto $T_{\mu(P_0)}\mathrm{aff}(\Bpure)$ at some pure state $P_0 \in (\densmap|_{\puremf(\H)})^{-1}(\rho_0)$. Then $\rho_0$ has an open neighborhood in $\mathrm{aff}(\Bpure)$ that is contained in $\densmap(\puremf(\H))$, implying that $\rho_0 \in \relint \densmap(\puremf(\H))$.
\end{proof}

Since they are critical values, we can thus describe boundary values always as the density of a $\Psi_0$ that has $\iota(v)\Psi_0=\lambda\Psi_0$. \citet[Cor.~4.2]{Gutkin2004} give the same result for the boundary of joint numerical ranges. This classification of critical values, as coming from eigenstates of operators $\iota(v)$, allows for the following more detailed description.

\begin{definition}\label{def:crit-plane}
    For every $v \in \V$ with $\iota(v) \not\propto I_\H$ the set 
    $\Pi^{(\ell)}_v=\{\densmap(P_\Psi) : P_\Psi\in\puremf(\H), \iota(v)\Psi=\lambda^{(\ell)}\Psi \}$
    with $\lambda^{(\ell)}$ being the $\ell$-th eigenvalue is called a \emph{(level-$\ell$) critical plane} created by $v$.
\end{definition}


The study of critical planes serves a dual purpose. First, they yield all possible critical points that will take a special role in some results formulated in Section~\ref{sec:functionals}. Second, since all boundary values are critical, the critical planes also form the boundary $\relboundary\Bpure$ and thus fully describe the set $\Bpure$. The last proposition of this section gives additional properties of critical planes that may be used to obtain a better understanding of the pure observable range $\Bpure$, and thereby the domain of the universal functional introduced in Section~\ref{sec:functionals}.


\begin{proposition}
    The following properties hold for critical planes.
    \begin{enumerate}[(i)]
        \item\label{item:critical-planes} The critical value set $\Bcrit$ is the union of all critical planes. 
        \item\label{item:crit-plane-perp} Every critical plane is a subset of an affine space in $\V^*$ that is perpendicular to the potential $v$ that creates it (justifying the name `plane').
        \item\label{item:crit-plane-intersection} For two level-$0$ critical planes $\Pi_{v_1}^{(0)}$, $\Pi_{v_2}^{(0)}$, created by $v_1$ and $v_2$, and for any $v=\alpha_1 v_1+\alpha_2 v_2$ with $(\alpha_1,\alpha_2)\in\R_{\geq 0}^2$,
        it holds that $\Pi_{v_1}^{(0)} \cap \Pi_{v_2}^{(0)}\subseteq \Pi_{v}^{(0)}$.
    \end{enumerate}
\end{proposition}

\begin{proof}
    \ref{item:critical-planes} 
    This follows directly from the definitions.
    \\
    \ref{item:crit-plane-perp} By Definition~\ref{def:crit-plane}, two points $\rho_1,\rho_2$ on the same critical plane must come from (normalized) eigenvectors of $v$ with the same eigenvalue, so $\iota(v)\Psi_1=\lambda\Psi_1$ and $\iota(v)\Psi_2=\lambda\Psi_2$. We have
    \begin{align}
        &\braket{\rho_1, v} = 
        \braket{\densmap(P_{\Psi_1}), v} = 
        \braket{\Psi_1, \iota(v)\Psi_1} = \lambda\quad\text{and}\quad\braket{\rho_2, v} = \braket{\densmap(P_{\Psi_2}), v} = \braket{\Psi_2, \iota(v)\Psi_2} = \lambda.
    \end{align}
    Taking the difference yields $\langle\rho_1-\rho_2,v\rangle=0$. So all points of the critical plane lie on an affine space perpendicular to $v$.\\
    \ref{item:crit-plane-intersection} 
    For the proof, consider the scope $\mathcal{S}_0:=(\V, \mathcal{H}, \iota, 0)$. That is, 
    we take the fixed part $H_0$ of the scope to be zero.
    Then, for any $v\in \V$, 
    the level-$0$ critical plane $\Pi_{v}^{(0)}$ is nothing but the set of ground state
    densities of the operator $\iota(v)$.
    Suppose $\rho \in \Pi^{(0)}_{v_1}\cap \Pi^{(0)}_{v_2}$. 
    Applying the weak Hohenberg-Kohn result (Theorem~\ref{th:weakHK} in Section~\ref{sec:functionals}; the proof does not depend on the present proposition) to $v_1, v_2$, and $\rho$,
    we can find a pure state $P_{\Psi}\in \puremf(\H)$, which is a ground state of both $\iota(v_1)$ and $\iota(v_2)$. That is,
    \begin{equation}
        \begin{aligned}
            \iota(v_1)\Psi=\lambda_1\Psi \quad\text{and}\quad
            \iota(v_2)\Psi=\lambda_2\Psi,
        \end{aligned}
    \end{equation}
    where $\lambda_1$ and $\lambda_2$ are the least eigenvalues of $\iota(v_1)$ and $\iota(v_2)$ respectively.
    It follows that
    \begin{equation}
        \iota(\alpha_1v_1+\alpha_2v_2) \Psi = (\alpha_1\lambda_1+\alpha_2\lambda_2)  \Psi.
    \end{equation}
    $\alpha_1\lambda_1 + \alpha_2\lambda_2$ is the least eigenvalue of $\iota(\alpha_1v_1+\alpha_2v_2)$
    since $\alpha_1,\alpha_2\ge 0$. This shows that $\rho = \mu(P_\Psi) \in \Pi^{(0)}_{\alpha_1v_1+\alpha_2v_2}$.
\end{proof}

\section{Examples for geometries arising from different scopes}
\label{sec:ex-geom}

In this section, we give several examples of functional-theory scopes and
discuss the corresponding problems of determining the observable ranges $\Bpure,
\Bens$. As these two sets do not depend on the choice of $H_0$ in the scope, we
will often write $\S=(\V,\H,\iota,\ast)$ to indicate that the choice of $H_0$ is arbitrary.

\subsection{Functional theory for fermions on a lattice}
\label{sec:lattice-geom}

Consider the Hilbert space $\H_1=\C^M$ for a single spinless particle on a lattice with $M$ sites. For $N<M$, the corresponding fermionic $N$-particle the Hilbert space is given by the $N$-fold exterior product, $\H_N=\wedge^N\H_1$. We assume $1\le N < M$, since otherwise the resulting physical system would be trivial. 
Let $a_i^*$ and $a_i$ denote the fermionic creation and annihilation operators at site $i$ (see Ref.~\cite{penz2021-Graph-DFT} for a detailed investigation of this setting). Note that $\sum_{i=1}^M a_i^*a_i=N\cdot I_{\H_N}$, so these operators are \emph{not} linear independent, resulting in the redundancy mentioned in the remark after Definition~\ref{def:scope}. 
To avoid this, we limit ourselves to $i=1,\ldots,M-1$, yielding the scope
\begin{equation}
    \S=(\R^{M-1},\H_N,\iota,\ast),
\end{equation}
with $\iota(e_i)=a_i^*a_i$.
Since the $a_i^*a_i$ all commute, $\S$ is an abelian scope. 
The density map applied to a pure state $P_\Psi$ evaluates to
\begin{equation}
    \mu(P_\Psi)=(\langle a^*_ia_i\rangle_\Psi)_{i=1}^m\in\R^m.
\end{equation}
Since the $a_i^*a_i$ commute, we can find a joint orthonormal eigenbasis $(\Phi_k)_{k=1}^d$ for $\H_N$. Here, the basis vectors are the states where $N$ out of $M$ sites are fully occupied while the others remain empty, which leads to $d={M \choose N}$ basis vectors. 
The corresponding $\densmap(P_{\Phi_k})$ are thus vectors in $\R^m$ with $N$ or $N-1$ ones and the rest filled with zero, where the remaining expectation value $\braket{\Phi_k, a^*_Ma_M\Phi_k}$ can always be determined from the normalization condition. Geometrically, this makes $\Bpure = \Bens$ a convex polytope with vertices $\densmap(P_{\Phi_k})$, as demonstrated already in the proof of Proposition~\ref{prop:ranges}\ref{prop:ranges:item:commute}, more specifically an $(M,N)$-hypersimplex (intersection of the hypercube $[0,1]^M$ with the $(M-1)$-dimensional hyperplane $\rho_1+\ldots+\rho_M=N$, see also Ref.~\cite[Sec.~II.C]{penz2021-Graph-DFT}). 
In this example, the faces of the convex polytope are critical planes (Definition~\ref{def:crit-plane}).

\subsection{Bloch sphere and disc}
\label{sec:Bloch}

This simple example already shows certain features of a non-abelian scope. We take $\H=\C^2$ and the Pauli matrices $\{X,Y,Z\}$ as basic observables.
The scope is $\S=(\R^3,\H,\iota,\ast)$, where $\iota(e_1)=X,\iota(e_2)=Y,\iota(e_3)=Z$, making $\iota$ injective.
Then for any pure state $P=P_\Psi\in\puremf(\H)$ we have
\begin{equation}\label{eq:Bloch-sphere}
    |\densmap(P)|^2 = \braket{X}_\Psi^2+\braket{Y}_\Psi^2+\braket{Z}_\Psi^2 = 1,
\end{equation}
where we use the shorthand notation $\braket{X}_\Psi=\langle \Psi,X\Psi \rangle$ etc.
The pure-state observable range is thus the Bloch sphere,
\begin{equation}
    \Bpure = \{ \rho\in\R^3 : |\rho|^2=1 \},
\end{equation}
and its convex hull $\Bens$ is the unit ball. This is also the simplest example for a non-convex $\Bpure$.
Note that since $\Bpure$ has empty relative interior, we have $\Bpure=\Bcrit$ in this case.

By removing one of the observables and restricting oneself to, say, $\{X,Y\}$, and thus reducing the scope, the constraint acts less restrictive and
\begin{equation}
    |\densmap(P)|^2 = \braket{X}_\Psi^2+\braket{Y}_\Psi^2 \leq 1.
\end{equation}
We thus get the observable ranges (``Bloch disc'')
\begin{equation}
    \Bpure = \Bens = \{ \rho\in\R^2 : |\rho|^2\leq 1 \}.
\end{equation}
Now $\Bpure$ is convex too and we get $\Bens=\Bpure$, while the critical values are the boundary of the disc, $\Bcrit=\relboundary\Bpure$.

In both cases, the Bloch sphere and disc, we have different eigenstates of $\iota(v)$ for every choice of $v\in\V\setminus\{0\}$ but always the eigenvalues $\pm 1$. This means that we have an infinite number of critical `planes', each one containing just a single point, exactly the points on the Bloch sphere (or on the boundary of the disc).

\subsection{Double qubit}
\label{sec:geometry_double_qubit}

Next, we take a double qubit system with $\H=\C^2\otimes\C^2$ and observables $\{X_1,Y_1,Z_1,Z_2\}$, where $X_1=X\otimes I_{\C^2}$ acts only on the first qubit, etc.
The scope is then $\S=(\R^4,\H,\iota,\ast)$, where $\iota(e_1)=X_1,\iota(e_2)=Y_1,\iota(e_3)=Z_1,\iota(e_4)=Z_2$.
We get the constraints 
\begin{align}
    \label{eq:dbl-qubit-norm-square}
    \braket{X_1}_\Psi^2+\braket{Y_1}_\Psi^2+\braket{Z_1}_\Psi^2 &\leq 1\quad\text{and}\quad
    \braket{Z_2}_\Psi^2 \leq 1,
\end{align}
where contrary to Eq.~\eqref{eq:Bloch-sphere} the first expression also allows values smaller than 1 due to additional freedom in the second qubit. This can be seen as follows. Take $\{\ket{\up}, \ket{\dn}\}$ as the eigenbasis of $Z$ for the single-spin Hilbert space, then any pure state in $\H$ can be written as a superposition $\Psi=c_+\varphi_+\otimes\ket{\up}+c_-\varphi_-\otimes\ket{\dn}$ with normalized $\varphi_+,\varphi_-\in\C^2$ and $|c_+|^2+|c_-|^2=1$.
A good example is the Bell state $\Psi=(\ket{\up\up}+\ket{\dn\dn})/\sqrt{2}$ that has $\braket{X_1}_\Psi=\braket{Y_1}_\Psi=\braket{Z_1}_\Psi=0$ and also $\braket{Z_2}_\Psi=0$.
For an arbitrary $\Psi\in\H$ like before we then get for the  expectation value of any $A\in\{X_1,Y_1,Z_1,Z_2\}$ that
\begin{equation}
    \braket{A}_\Psi = |c_+|^2\braket{A}_{\varphi_+}+|c_-|^2\braket{A}_{\varphi_-},
\end{equation}
which makes it a convex combination of values for states $\Psi_+=\varphi_+\otimes\ket{\up}$ and $\Psi_-=\varphi_-\otimes\ket{\dn}$ with fixed $\braket{Z_2}_{\Psi_+}=+1$ and $\braket{Z_2}_{\Psi_-}=-1$. But those extremal states, with the second qubit frozen at $\ket{\up}$ ot $\ket{\dn}$, must lie on the Bloch sphere like before,
\begin{align}
    \braket{X_1}_{\Psi_+}^2+\braket{Y_1}_{\Psi_+}^2+\braket{Z_1}_{\Psi_+}^2 = 1 \;\text{at}\; \braket{Z_2}_{\Psi_+}=+1,\\
    \braket{X_1}_{\Psi_-}^2+\braket{Y_1}_{\Psi_-}^2+\braket{Z_1}_{\Psi_-}^2 = 1 \;\text{at}\; \braket{Z_2}_{\Psi_-}=-1.
\end{align}
So the following situation unfolds for a density from the pure-state observable range, $\rho\in\Bpure$. We have two unit spheres (Bloch spheres) in $(\rho_1,\rho_2,\rho_3)$ at coordinates $\rho_4=\braket{Z_2}_\Psi=\pm 1$. To this we add all convex combinations between two points from each of the spheres. We illustrate this situation by cutting $\Bpure\subseteq\R^4$ at $\rho_3=\braket{Z_1}_\Psi=0$ in Figure~\ref{fig:two-qubit-ex}.
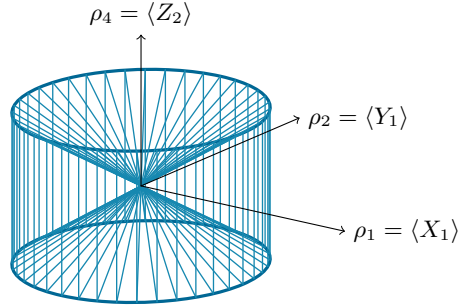
\begin{figure}
    \centering
    \begin{tikzpicture}[line cap=round, line join=round, 
        x={(0.9cm,-0.2cm)}, y={(0.7cm,0.3cm)}, z={(0,1cm)},
        draw=MidnightBlue]
    \xdef\N{40} 
    \xdef\D{1.5}
    \begin{scope}[canvas is xy plane at z=1]
        \draw[line width=1.2pt] (0,0) ellipse (\D cm and \D cm);
        \foreach \i in {1,...,\N} {
            \pgfmathsetmacro{\ang}{((\i-1)*360/\N)};
            \coordinate (top\i) at (\ang:\D cm and \D cm);
        }
    \end{scope}
    \begin{scope}[canvas is xy plane at z=-1]
        \draw[line width=1.2pt] (0,0) ellipse (\D cm and \D cm);
        \foreach \i in {1,...,\N} {
            \pgfmathsetmacro{\ang}{((\i-1)*360/\N)};
            \coordinate (bottom\i) at (\ang:\D cm and \D cm);
        }
    \end{scope}
    \foreach \i in {1,...,\N} {
        \draw[color=MidnightBlue!70,line width=.5pt] (top\i) -- (bottom\i);
        \pgfmathsetmacro{\j}{\i+\N/2}; 
        \pgfmathsetmacro{\j}{ifthenelse(\j>\N, \j-\N, \j)}
        \draw[color=MidnightBlue!70,line width=.5pt] (top\i) -- (bottom\j);
    }
    \draw[->, black, thin] (0,0,0) -- (3,0,0) node[right] {$\rho_1=\braket{X_1}$};
    \draw[->, black, thin] (0,0,0) -- (0,3,0) node[right] {$\rho_2=\braket{Y_1}$};
    \draw[->, black, thin] (0,0,0) -- (0,0,2) node[above] {$\rho_4=\braket{Z_2}$};
    
    \end{tikzpicture}
    \caption{
        The set $\Bpure$ from the double qubit example consists of two Bloch spheres in $(\rho_1,\rho_2,\rho_3)$ at $\rho_4=\pm 1$ and all convex combinations between two points from each sphere. Leaving out the coordinate $\rho_3$ in this illustration makes the spheres appear as circles on the top and bottom. Setting $\rho_3=0$, the resulting shape is a cylinder from which two cones are removed at the top and bottom. Its convex hull, the set $\Bens$, is then the full cylinder.
    }
    \label{fig:two-qubit-ex}
\end{figure}
This examples shows the possibility of critical values inside $\relint\Bens$, exactly the cone-shaped set of all convex combinations from opposite points on the two Bloch spheres.

\subsection{Momentum functional theory for bosons on a periodic lattice}
\label{sec:momentum-ft}

Here, we illustrate how momentum-space functional theory for bosons \cite{LiebertBoseEinstein2021,WangBosonsRDMFT2026} fits into our general framework.
Consider a system of $N$ spinless bosons on a periodic lattice with $L^d$ sites, where $d$ is the spatial dimension. A common situation is that one can physically tune the hopping rate, changing the kinetic part of the Hamiltonian~\cite{Guenter2013,Schempp2015}. 
Let $\H$ be the $N$-particle Hilbert space, and consider the momentum space number operators
\begin{equation}
n_{\bm{k}}\in \Lsa(\H) \quad\text{with}\;\bm{k}\in (\mathbb{Z}/L\mathbb{Z})^d.
\end{equation}
Let $\V = \R^{(\mathbb{Z}/L\mathbb{Z})^d}$, and consider the parameterization map
\begin{equation}
\begin{aligned}
 \iota: &\;\V \rightarrow \Lsa(\H) \\
& v\mapsto \sum_{\bm{k}\in(\mathbb{Z}/L\mathbb{Z})^d} v_{\bm{k}} n_{\bm{k}}.
\end{aligned}
\end{equation}
We will consider the scope $\mathcal{S} = (\V, \H, \iota, H_0)$, where $H_0$ is a fixed boson-boson interaction.
The density map for $\mathcal{S}$ is then given by
\begin{equation}
\begin{aligned}
\mu : \densmat(\H) &\rightarrow \mathbb{R}^{(\mathbb{Z}/L\mathbb{Z})^d}\\
\Gamma &\mapsto (\trace(n_{\bm{k}} \Gamma))_{\bm{k}\in (\mathbb{Z}/L(\mathbb{Z})^d}.
\end{aligned}
\end{equation}
In the presence of translation invariance, we can restrict the Hilbert space to the subspace $\mathcal{H}_{\bm{p}}$ with fixed total momentum $\bm{p}\in (\mathbb{Z}/L\mathbb{Z})^d$. In other words, the functional theory for a fixed-momentum sector is given by the scope $(\V, \mathcal{H}_{\bm{p}}, \iota, H_0)$, where $H_0$ is a fixed translationally invariant interaction operator, and the operators $\iota(v)$ and $H_0$ are understood as acting on $\mathcal{H}_{\bm{p}}$.
Since the $n_{\bm{k}}$ commute with each other, this gives an abelian scope and we have from Proposition~\ref{prop:ranges}\ref{prop:ranges:item:commute} that
\begin{equation}
\begin{aligned}
\Bpure = \Bens =\conv\left\{
\mu(P_\Psi) : 
\Psi\in\mathcal{H}_{\bm{p}} \text{ is simultaneous eigenvector of all } n_{\bm{k}} \right\}\subseteq  \mathbb{R}^{(\mathbb{Z}/L\mathbb{Z})^d}.
\end{aligned}
\end{equation}
Note that the eigenvalues of $n_{\bm{k}}$ are always integer. Clearly, a density $\rho=\mu(P_\Psi)$ is the simultaneous eigenvalue of all $n_{\bm{k}}$ if and only if 
\begin{equation}
\begin{aligned}
& \sum_{\bm{k}\in (\mathbb{Z}/L\mathbb{Z})^d}  \rho_{\bm{k}}\bm{k} = \bm{p} \;(\text{in }(\mathbb{Z}/L\mathbb{Z})^d),\quad \sum_{\bm{k}\in (\mathbb{Z}/L\mathbb{Z})^d} \rho_{\bm{k}} = N,\quad\text{and}\quad \rho_{\bm{k}} \ge 0.
\end{aligned}
\end{equation}
For example, if $d=2, L=2, N=2,$ and $\bm{p}=(1,0)$, there are two $\rho$, given by
\begin{equation}
\rho^{(1)} = (1,0,1,0), \quad \rho^{(2)} = (0,1,0,1),
\end{equation}
where we flatten a vector $\rho \in \mathbb{R}^{(\mathbb{Z}/2\mathbb{Z})^2} \cong \mathbb{R}^4$ as $\rho = (\rho_{(0,0)}, \rho_{(0,1)}, \rho_{(1,0)}, \rho_{(1,1)})$. Hence, we have $\Bpure = \Bens = \{(t, 1-t, t, 1-t) : t \in [0,1]\}$, a line segment in $\mathbb{R}^4$.
In Fig.~\ref{fig:bosonic_d1L3N5}, we show the observable ranges $\Bpure=\Bens$ with $N=4$ bosons on a one-dimensional lattice of $L=3$ sites. For $p=0$, for example, the vectors of simultaneous eigenvalues are $(4,0,0), (2,1,1), (1,3,0), (1,0,3),$ and $(0,2,2)$. The sets $\Bpure = \Bens$ are then the convex hull of these five vectors in $\mathbb{R}^3$.
For each $p$, since $\Bpure=\Bens$ is two-dimensional, the critical values 
are exactly those that lie on the convex hull of at most two $n_{\bm k}$'s.

\begin{figure}
\centering
\includegraphics[width=.5\columnwidth]{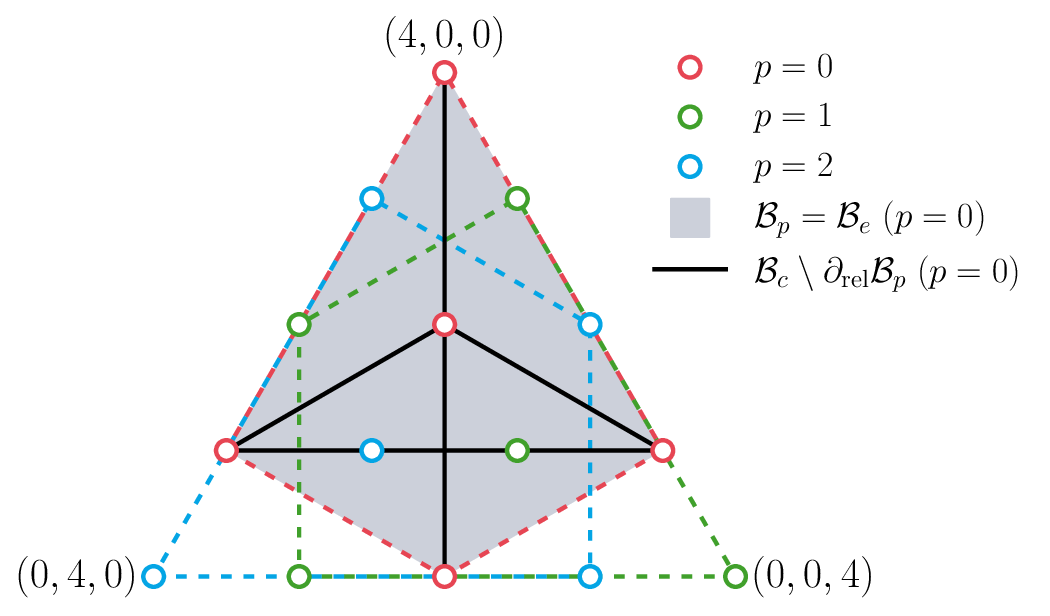}
\caption{The pure-state and ensemble observable ranges $\Bpure=\Bens$ of the one-dimensional bosonic momentum-functional theory for lattice size $L=3$, particle number $N=4$, and total momentum $p=0,1,2$. For $p=0$, $\Bpure=\Bens$ is colored in gray, and the set of critical values 
$\mathcal{B}_c$ excluding the relative boundary is indicated in black.
}
\label{fig:bosonic_d1L3N5}
\end{figure}

\section{Constrained-search functionals and representability}
\label{sec:functionals}

The central objects of interest of a functional theory are the so-called universal functionals. 
These are defined as the expectation value of $H_0$
minimized over either the set of pure states or the set of ensemble states constrained on a given density.
For a given density $\rho\in\V^*$, we write $\Psi\mapsto\rho$ for $\densmap(P_\Psi)=\rho$ and $\Gamma\mapsto\rho$ for $\densmap(\Gamma)=\rho$.

\begin{definition}\label{def:const-search-func}
    For a given scope $\S=(\V,\H,\iota,H_0)$ we define the pure-state and ensemble constrained-search functionals:
    \begin{alignat}{2}
        \label{eq:def-Fpure}
        &\Fpure^{(\S)}(\rho) := &\left\{ \begin{array}{ll}
        \inf\limits_{\Psi\mapsto\rho} \langle\Psi,H_0\Psi\rangle & \text{if}\; \rho\in\Bpure\\
        +\infty & \text{else,}
        \end{array} \right.\\
        \label{eq:def-Fens}
        &\Fens^{(\S)}(\rho) := &\left\{ \begin{array}{ll}
        \inf\limits_{\Gamma\mapsto\rho} \trace(H_0\Gamma) \,\,& \text{if}\; \rho\in\Bens\\
        +\infty & \text{else}.
        \end{array} \right.
    \end{alignat}
\end{definition}

Whenever the scope $\S$ is clear from the context, we will omit the superscript `$(\S)$' and simply write $\Fpure$ and $\Fens$.
These definitions can be used directly to find an expression for the ground-state energy, where it does not matter if one varies over pure or ensemble states,
\begin{equation}\label{eq:E-def}
\begin{aligned}
    E(v) &= \inf_{P\in\puremf(\H)}\trace(H(v)P) \hspace{-.5em}&= \inf_{\rho\in\Bpure}\{ \Fpure(\rho) + \langle \rho,v\rangle \} \\
    &= \inf_{\Gamma\in\densmat(\H)}\trace(H(v)\Gamma) &= \inf_{\rho\in\Bens}\{ \Fens(\rho) + \langle \rho,v\rangle \}.
\end{aligned}
\end{equation}
Since $\puremf(\H) \subseteq \densmat(\H)$ are compact, these infima are always attained.
That is, the ground-state energy functional $E:\V \rightarrow\mathbb{R}$ is (up to a sign choice) the Legendre--Fenchel transform (convex conjugate) of $\Fpure$ and $\Fens$.
For a fixed external potential $v$, 
by the Rayleigh--Ritz variational principle, any minimizer $P=P_\Psi\in\mathbb{P}(\H)$ or $\Gamma\in \densmat(\H)$ in the expressions above is a ground state 
of the Hamiltonian $H_0+\iota(v)$.
Since the functionals have been defined with value $+\infty$ outside of their effective domains $\Bpure$ and $\Bens$, the domain for the variation can in general be extended to all of $\V^*$ and we thus have
\begin{equation}
    E(v) = \inf_{\rho\in\V^*}\{ \Fpure(\rho) + \langle \rho,v\rangle \} = \inf_{\rho\in\V^*}\{ \Fens(\rho) + \langle \rho,v\rangle \}.
\end{equation}

There are several possibilities to numerically evaluate the functionals $\Fpure$
and $\Fens$. The obvious way is to solve the respective variational problem from
the definition directly, but there are other means that sometimes prove to be
more practical. 
For the ensemble constrained-search functional
$\Fens$, 
Proposition~\ref{prop:cs-func}\ref{prop:cs-func:item:LF} below defines a
dual problem that when solved yields not only the value of the functional but
also a representing $v\in\V$ for the given density $\rho$. In DFT, this
important method is called ``Lieb
optimization''~\cite{Lieb-Book-Trygve,Garrigue2022}, because it builds on the
mathematical treatment of DFT in terms of convex functionals pioneered by
\citet{Lieb1983}. 
This method is part of a larger collection of 
``inverse Kohn--Sham''
methods~\cite{vanleeuwen1994exchange,Zhao1994,Shi2021,Penz2023-MY-ZMP}, where it
is usually easier to find $\Fens$, due to its convexity, rather than the pure-state functional. 
There is, however, also the possibility to use imaginary-time evolution of the
wave function to compute $\Fpure$~\cite{Penz2025-cs-imag-time}. 
A fully algebraic method for $\Fpure$ that solves the underlying polynomial
equations from the constraints together with the eigenvalue equation from the
Hamiltonian, but thus also needs to include excited states, was given by
\citet{Song2023}.
We note that these exact methods for computing $\Fpure$ or $\Fens$ are largely
limited to small systems and are mainly useful for studying the theoretical
properties of universal functionals. For practical applications, it is
sometimes possible to employ wave function 
ansatzes to reduce the complexity of the constrained search. 
For example, in RDMFT, the use of the 
APSG ansatz results in the PNOF5 functional \cite{pirisNaturalOrbitalFunctional2011,pernalEquivalencePirisNatural2013}.

\subsection{Properties of constrained-search functionals}

\begin{proposition}\label{prop:cs-func}
    The following properties hold for the constrained-search functionals.
    \begin{enumerate}[(i)]
        \item\label{prop:cs-func:item:inf-attained} 
        In the definitions of both functionals, i.e., Eq.~\eqref{eq:def-Fpure} and Eq.~\eqref{eq:def-Fens}, the infimum is attained.
        \item\label{prop:cs-func:item:ens-convex} $\Fens$ is convex.
        \item\label{prop:cs-func:item:ens-lsc} $\Fens$ is lower semicontinuous on $\V^*$ and continuous on $\relint\Bens$.
        \item\label{prop:cs-func:item:pure-lsc} $\Fpure$ is lower semicontinuous on 
        $\V^*$ and continuous on $\Breg$.
        \item\label{prop:cs-func:item:conv-envelope} $\Fens$ is the lower convex envelope of $\Fpure$, in particular $\Fens\leq\Fpure$.
        \item\label{prop:cs-func:item:LF} $\Fens(\rho)=\sup_{v\in\V}\{E(v)-\langle \rho,v\rangle\}$ for all $\rho\in\Bens$.
        \item\label{prop:cs-func:item:flat} If $H_0=\iota(h_0)$ with $h_0\in\V$, then $\Fpure(\rho)=\langle\rho,h_0\rangle$ for all $\rho\in\Bpure$ and $\Fens(\rho)=\langle\rho,h_0\rangle$ for all $\rho\in\Bens$
        . In particular, the functionals are affine on the respective domains.
        \item\label{prop:cs-func:item:piecewise-flat} If the functional theory is abelian and further $[H_0,\iota(v)]=0$ for every $v\in\V$, then $\Fpure=\Fens$ is a continuous, convex, and piecewise affine function on $\Bpure=\Bens$.
    \end{enumerate}
\end{proposition}

For the last point, recall that a real function defined on a polyhedral subset of a real vector space is called piecewise affine if its domain can be decomposed into polyhedral subsets again, on each of which the function is affine.

\begin{proof}
    \ref{prop:cs-func:item:inf-attained} 
     This follows from the fact that the sets $\{P\in\puremf(\H):\densmap(P)=\rho\}$ and $\{\Gamma\in\densmat(\H):\densmap(\Gamma)=\rho\}$ are both compact. Indeed, each of these sets is the preimage of a point under a continuous map, and thus closed. Since $\puremf(\H) \subseteq \densmat(\H)$ are compact, the sets are also compact.\\
    \ref{prop:cs-func:item:ens-convex} Let $\rho = {\lambda} \rho_1 + (1-\lambda)\rho_2\in\Bens$, $\lambda\in(0,1)$. Then $\densmap(\Gamma_1)=\rho_1$ and $\densmap(\Gamma_2)=\rho_2$ implies $\densmap(\lambda\Gamma_1+(1-\lambda)\Gamma_2) =\rho$ (but not the other way around) since $\densmap$ is linear on density matrices.
    We then have
    \begin{equation}\begin{aligned}
        \lambda\Fens(\rho_1)+(1-\lambda)\Fens(\rho_2) &= \inf_{\Gamma_1\mapsto\rho_1,\Gamma_2\mapsto\rho_2} \trace(H_0(\lambda\Gamma_1+(1-\lambda)\Gamma_2)) \\
        &\geq \inf_{\lambda\Gamma_1+(1-\lambda)\Gamma_2\mapsto\rho}\trace(H_0(\lambda\Gamma_1+(1-\lambda)\Gamma_2)) = \Fens(\rho),    
    \end{aligned}\end{equation}
    where the inequality follows from the second constraint for $\Gamma_1$ and $\Gamma_2$ being more general. This shows convexity.\\
    \ref{prop:cs-func:item:ens-lsc}
    This proof is a finite-dimensional adaption from \citet[Th.~4.4]{Lieb1983}.
    Take $\rho_i\to\rho$ for an $\rho\in\Bens$ then we need to show $\liminf_{i\rightarrow\infty}\Fens(\rho_i) \geq \Fens(\rho)$ to prove lower semicontinuity (lsc) on $\Bens$.
    We can assume all $\rho_i\in\Bens$ as well, else $\Fens(\rho_i)=+\infty$ and there is nothing to show.
    We also already choose a subsequence $(\rho_{i'})_{i'}$ that realizes the limes inferior as $\lim_{i'\rightarrow\infty}\Fens(\rho_{i'})$.
    Without loss of generality, replace $H_0$ by the positive operator $\Hpos = H_0 +cI_\H$ by just choosing a large enough $c\in\R$. Since the infimum in the definition of $\Fens$ is always attained by Proposition~\ref{prop:cs-func}\ref{prop:cs-func:item:inf-attained}, this allows us to choose a $\Gamma_i\mapsto\rho_i$ that satisfies $0\leq \trace(\Hpos\Gamma_i)=\Fens(\rho_i)\leq \|\Hpos\|$ for every $i$. Now the trace norm is equivalent to the operator norm in finite dimensions, so $\Hpos\Gamma_i$ is uniformly bounded and we can find a convergent subsequence $\Hpos\Gamma_{i''}\to A$ that is already selected from the subsequence $\Gamma_{i'}\mapsto\rho_{i'}$ that realizes the limes inferior. Since $H_+$ is invertible, $\Gamma_{i''}\to \Gamma=\Hpos^{-1}A$, and from $\densmap$ continuous it follows $\lim_{i''\to\infty}\densmap(\Gamma_{i''})=\densmap(\Gamma)=\rho$. From the definition of the constrained-search functional it clearly follows that $\Fens(\rho)\leq \trace(\Hpos\Gamma)$. We can now summarize
    \begin{equation}
        \liminf_{i\rightarrow\infty}\Fens(\rho_i) = \lim_{i''\rightarrow\infty}\Fens(\rho_{i''}) = \trace(\Hpos\Gamma) \geq \Fens(\rho),
    \end{equation}
    which shows lsc on $\Bens$. The functional is then also lsc on all of $\V^*$ since $\Bens$ is closed and it is extended to $\V^*$ by setting it to $+\infty$. That a convex lsc function is continuous on the relative interior of its effective domain is a standard result in finite dimensions~\cite[Th.~10.1]{Rockafellar-book-1970}.\\
    \ref{prop:cs-func:item:pure-lsc}
    Let $\rho \in \Bpure$ and let $(\rho_i)_i$ be a sequence in $\Bpure$
    such that $\rho_i \rightarrow \rho$. 
    Pick a subsequence $(\rho_{i'})_{i'}$ such that $\Fpure(\rho_{i'}) \rightarrow \liminf_{i\rightarrow \infty} \Fpure(\rho_i)$.
    For each $\rho_{i'}$, pick a pure state $P_{i'} \in \puremf(\H)$ such that
    $P_{i'} \mapsto \rho_{i'}$ and $P_{i'}$ is a minimizer of the constrained search at $\rho_{i'}$. 
    By compactness, we can extract a subsequence $(P_{i''})_{i''}$
    such that $P_{i''}\rightarrow P$ for some $P\in \puremf(\H)$. By continuity, $\mu(P) = \rho$ and 
    \begin{equation}
    \trace(PH_0) = 
    \lim_{i''\rightarrow\infty} \trace(P_{i''}H_0) 
    = \lim_{i''\rightarrow\infty} \Fpure(\rho_{i''}) 
    = \liminf_{i\rightarrow\infty} \Fpure(\rho_i).
    \end{equation}
    In other words, $\liminf_{i\rightarrow \infty}\Fpure(\rho_i) \ge \Fpure(\rho)$.
    Note that this argument also works for $\Fens$.\\
    We now show the continuity of $\Fpure$ on $\Breg$.
    Take $\rho \in \Breg$. Since $\Breg$ is open, we can find an open neighborhood $U\ni \rho$ contained in $\Breg$. Let $M:= \densmap^{-1}(U) \subseteq \puremf(\H)$, which is an open submanifold.
    Then $\densmap|_M: M \rightarrow U$ is a proper surjective submersion, 
    so by Ehresmann fibration theorem~\cite[Th.~5.5.14]{abraham-book},
    shrinking $U$ if necessary, 
    there exists a manifold $F$ and a diffeomorphism 
    $\phi: M \rightarrow U\times F$ such that $\pi_1\circ \phi = \densmap|_M$,
    where $\pi_1: U\times F\rightarrow U$ denotes projection onto the first factor.
    Note that $F$ is compact because it is diffoemorphic to $\mu|_M^{-1}(\rho)$.
    Thus, on $U$, we have $\Fpure(\rho) = \min_{f\in F} \trace(\phi^{-1}(\rho, f) H_0)$. Since $F$ is compact and $(\rho, f)\mapsto \trace(\phi^{-1}(\rho, f)H_0)$ is continuous, so is $\Fpure|_U$.
    \\
    \ref{prop:cs-func:item:conv-envelope} 
    Restating Eq.~\eqref{eq:E-def}, we have $E(v) = -\Fens^*(-v) = -\Fpure^*(-v)$, where ${}^*$ denotes the usual Legendre--Fenchel transformation. Since $\Fens$ is convex, proper, and lower semicontinuous, we have $\Fens = \Fens^{**} = \Fpure^{**}$ (biconjugate). But double Legendre--Fenchel transformation (of $\Fpure$) exactly corresponds to the (closed) lower convex envelope~\cite[Th.~11.1]{rockafellar2009variational}.\\
    \ref{prop:cs-func:item:LF}
    This was already shown in \ref{prop:cs-func:item:conv-envelope} with $\Fens = \Fens^{**} = (-E(-\cdot))^*$.\\
    \ref{prop:cs-func:item:flat} For $H_0=\iota(h_0)$ and $\mu(\Psi)=\rho$ the expectation value in Eq.~\eqref{eq:def-Fpure} is $\langle\Psi, H_0\Psi\rangle = \langle\rho,h_0\rangle$ for $\rho\in\Bpure$ by Eq.~\eqref{eq:density-map}. The same holds for density matrices by Eq.~\eqref{eq:mu-DM} for $\rho\in\Bens$.\\
    \ref{prop:cs-func:item:piecewise-flat} Suppose $[\iota(v),\iota(w)]=0$ and $[H_0,\iota(v)]=0$ for all $v,w\in \V$. Note that in the abelian case $\Bpure=\Bens$ from Proposition~\ref{prop:ranges}\ref{prop:ranges:item:commute}. There is then an orthonormal basis $(\Phi_k)_{k=1}^{d}$ for $\H\cong\C^d$ such that $\iota(v) \Phi_k = \braket{\omega_k, v}\Phi_k$ for all $v\in \V$ 
    and, furthermore, $H_0\Phi_k=\lambda_k\Phi_k$.
    \begin{equation}
        \begin{aligned}
            &\Fpure(\rho) = \min_{\Psi \mapsto \rho} \braket{\Psi,H_0\Psi}
            = \min_{\Psi\mapsto\rho} 
            \sum_{k=1}^{d} \lambda_k \left|\braket{\Phi_k,\Psi}\right|^2= \min\left\{\sum_{k=1}^{d} t_k\lambda_k : \sum_{k=1}^{d}t_k = 1, t_k \ge 0, \sum_{k=1}^{d}t_k \omega_k = \rho\right\}.
        \end{aligned}
    \end{equation}
    This is a linear programming problem of the form
    \begin{equation}
    \min_t \braket{\lambda, t} \hspace{2em} t \ge 0,\; At = \rho\oplus 1,
    \end{equation}
    where $A$, defined from the $\omega_k$, is in $\mathrm{Hom}(\R^{d}, \V^*\oplus \R)$ (because the variable $t$ lives in $\R^d$). 
    The dual problem is
    \begin{equation}
    \max_{y\in \V \oplus \mathbb{R}} \braket{\rho\oplus 1, y}\hspace{2em} A^* y \le \lambda.
    \end{equation}
    The feasible set $\{y\in \V \oplus \mathbb{R} : A^* y \le \lambda\}$ is a 
    (convex) polyhedron, whose vertices we denote by $\nu^{(1)}, \ldots, \nu^{(N)}$.
    If the primal problem has a solution, i.e., if $\rho \in \Bpure=\Bens$, then
    by strong duality we have
    \begin{equation}
    \Fpure(\rho) = \max_{i\in \{1, \ldots, N\}} \braket{\rho\oplus 1, \nu^{(i)}}.
    \end{equation}
    That is, $\Fpure$ is the pointwise maximum of a finite collection of affine functions on $\Bpure$, implying that $\Fpure$ is continuous, convex, and piecewise affine on $\Bpure$. Hence, we also have $\Fpure = \Fpure^{**} = \Fens$.
\end{proof}

One might wonder if the continuity result from \ref{prop:cs-func:item:ens-lsc} cannot be shown including the boundary of $\Bens$. In Appendix~\ref{app:discont-ens}, we give a simple example of a functional $\Fens$ that is \emph{not} continuous on $\relboundary\Bens$. 
Unlike the ensemble functional $\Fens$, the pure-state functional $\Fpure$ does not even need to be continuous on the interior of its domain, for which we give an example in Appendix~\ref{app:discont-pure}. 
Lastly, Appendix~\ref{app:discont-real} provides an example of a pure-state functional that is lower semicontinuous but not continuous, whenwconstrained search is limited to real expansion coefficients.
As a further general feature, the universal functionals exhibit a divergent derivative close to the boundary of their effective domains. A careful study of this behavior is beyond the current work, but it was already described for different settings in Refs.~\cite{schillingDivergingExchangeForce2019,benavides-riverosReducedDensityMatrix2020,LiebertBoseEinstein2021,maciazekRepulsivelyDivergingGradient2021,Liebert2023-refining,WangBosonsRDMFT2026}.

\subsection{\texorpdfstring{$v$}{v}-representability}
Before proceeding, we recall a standard notion from convex analysis.
The \emph{subdifferential} of a convex functional $f$ at some $x\in X$ is the set-valued map~\cite[Th.~23.2]{Rockafellar-book-1970}
\begin{equation}\label{eq:subdiff}
    \partial f(x) = \{u\in X^*: \forall y \in X\colon f(y)\geq f(x) + \langle u,y-x \rangle\}.
\end{equation}
On the real line, the set $\partial f(x)$ can be interpreted as the collection of all slopes for which the corresponding tangent of $f$ at $x$ never exceeds the function $f$ itself.
For concave functions, 
the appropriate notion is
\emph{superdifferentials}, defined by reverting the inequality, and we use the same notation in this case.
Note that $\partial f(x)$ is always a convex and closed set, and that it includes infinitely many elements where $f$ has a kink and is thus non-differentiable. It can also be the empty set. In functional theories, an extremely useful application of the subdifferential lies in the statement $(\partial f)^{-1}=\partial f^*$ for convex, proper, and lower semicontinuous functionals. This translates to
\begin{equation}\label{eq:subdiff-rel}
    (-\partial\Fens)^{-1}=\partial E
\end{equation}
in the present context.

We are now in a position to answer an important question: whether for a given density $\rho$ there exists a potential $v\in\V$ such that the ground state solution of the Schrödinger equation with Hamiltonian $H(v)$ exactly yields $\rho$. 
This property is highly relevant for the formulation and application of functional theories, since the theory can only effectively work on the set of densities that are representable by some choice of $v\in\V$.
We then call $\rho$ a \emph{$v$-representable} density, and we find a two-fold answer for when this holds.
The question of uniqueness is later tackled in Theorem~\ref{th:strongHK}.

\begin{theorem}\label{th:v-rep}
    \begin{enumerate}[(i),wide=0pt,parsep=0pt,itemsep=0pt]
        \item\label{th:v-rep:item:pure} For every $\rho\in\Breg$, there is a $v\in\V$ and $\Psi\in\H$ such that $\Psi$ is an eigenstate of $H(v)$ and $\mu(P_\Psi) = \rho$ (\emph{pure eigenstate $v$-representability}).
        \item\label{th:v-rep:item:ens} For every $\rho\in\relint\Bens$ we find $v\in\V$ and $\Gamma\in\densmat(\H)$ such that $\Gamma$ is an ensemble ground state of $H(v)$, $\Gamma\mapsto\rho$ (\emph{ensemble ground-state $v$-representability}), and $-v\in\partial\Fens(\rho)$.
    \end{enumerate}
    \noindent
    In both cases this state is also the optimizer in the respective constrained-search functional.
\end{theorem}

\begin{proof}
\ref{th:v-rep:item:pure}
For all $\rho\in\Breg$ the set $\mathcal{M}_{\rho}:=\{P\in\puremf(\H) : \densmap(P)=\rho\}$ is an embedded manifold~\cite[Th.~3.5.4]{abraham-book}, over which the constrained minimization $\Fpure(\rho)=\inf_{P\in\mathcal{M}_{\rho}}\trace(H_0 P)$ is carried out. 
Let $P_\Psi\in\mathcal{M}_{\rho}$ be a minimizer of the constrained search and let $F: \puremf(\H) \rightarrow \R, P\mapsto \trace(PH_0)$
denote the corresponding objective function. By the Lagrange multiplier theorem~\cite{LM-arxiv,abraham-book} there exists a $v\in \V$
such that
\begin{equation}
\rmd_{P_\Psi}F + \braket{\rmd_{P_\Psi}\mu, v} = 0.
\end{equation}
Now, evaluate this on a tangent vector $[P_\Psi, S]\in T_{P_\Psi}\puremf(\H)$,
where $S\in \rmi\Lsa(\H)$
to get
\begin{equation}
\trace([P_\Psi, S] H_0) + \trace([P_\Psi, S]\iota(v)) = 0.
\end{equation}
Since this is true for all $S\in \rmi\Lsa(\H)$, by the cyclicity of trace we have
\begin{equation}
    [P_\Psi, H_0 + \iota(v)] = 0.
\end{equation}
This is equivalent to $(H_0 + \iota(v)) \Psi = H(v)\Psi = \lambda \Psi$ as shown in the proof of Lemma~\ref{lem:crit-commute}.


\ref{th:v-rep:item:ens}
We have that $\Fens$ has a non-empty (but possibly multi-valued) subdifferential for all $\rho\in\relint\Bens$~\cite[Th.~23.4]{Rockafellar-book-1970}, so choose $v\in-\partial\Fens(\rho)$. Then, by the definition of the subdifferential, $\rho$ minimizes $\Fens(\rho)+\langle\rho,v\rangle$. Let $\Gamma\mapsto\rho$ be a minimizer in Eq.~\eqref{eq:def-Fens} and $\Gamma'\mapsto\rho'$ another density matrix. Then
\begin{equation}
\begin{aligned}
    \trace(H(v)\Gamma') &= \trace(H_0\Gamma') + \langle\rho',v\rangle \geq \Fens(\rho') + \langle\rho',v\rangle\geq \Fens(\rho) + \langle\rho,v\rangle \\&= \trace(H_0\Gamma) + \langle\rho,v\rangle = \trace(H(v)\Gamma).
\end{aligned}
\end{equation}
So $\Gamma$ minimizes $\trace(H(v)\Gamma)$ and is thus a ground state.
\end{proof}

Such a proof for pure-eigenstate $v$-representability was originally devised for the setting of a Dicke model on an infinite-dimensional Hilbert spaces, where the pure-state $v$-representability is even demonstrated with low-lying eigenstates~\cite{Bakkestuen2025-Dicke}.
Another pure-eigenstate $v$-representability result was established with imaginary-time evolution in the setting of an abelian functional theory~\cite{Penz2025-cs-imag-time}, where an extension to the non-abelian case is also possible in principle.
Note further that in a one-dimensional continuum setting it is possible to prove analyticity of the mapping from densities to ground states~\cite{Corso-regularity-arxiv} that then also extends to analyticity of the pure-state constrained-search functional $\Fpure$. Yet this result depends on pure-state $v$-representability by \emph{ground states} and can thus not be extended to our setting. Results that establish ``excited-state $v$-representability'' like Theorem~\ref{th:v-rep}\ref{th:v-rep:item:pure} unexpectedly help to get a hold on properties of excited states only through the pure-state constrained-search functional~\cite{Perdew1985,Yang2024}. The classical proof for ensemble ground-state $v$-representability for abelian functional theories, even for infinite-dimensional discrete settings, is from \citet{CCR1985}.

Our result in Theorem~\ref{th:v-rep} leaves open the question of ensemble $v$-representability on the boundary of $\Bens$ and by pure states for critical values. Interestingly, it was found that at critical values a non-unique $v$-representability is possible when degeneracy occurs, thus yielding counterexamples to a possible Hohenberg--Kohn theorem~\cite{penz2021-Graph-DFT,penz2023geometry}. Conversely, the same phenomenon also allows certain boundary values to become $v$-representable in exceptional cases.

We would argue that the criteria for $v$-representability are the most important ingredients for the construction of a functional theory, since we want to replace ground states $\Psi$ and $\Gamma$ by their reduced quantity $\rho\in\V^*$ and thus need to work within the set of $v$-representable elements in $\V^*$.
While the answer given by Theorem~\ref{th:v-rep} is very satisfying in any discrete setting, the situation is very different for the setting of DFT in the continuum and the ``$v$-representability problem'' appears. There, explicit examples for non-$v$-representable densities can be constructed~\cite{ENGLISCH1983} and they are arbitrarily close to representable densities~\cite{Lammert2007,Garrigue2021}. To date, positive results for density spaces that allow full $v$-representability were only found in one-dimensional settings~\cite{Sutter2024,CarvalhoCorso2025}.

\subsection{Hohenberg--Kohn theorems}
\label{sec:hk_theorems}

After the discussion of $v$-representability, an ensuing question is then if a density admits a \emph{unique} representing potential that consequently allows the definition of a density-potential map $\rho\mapsto v$. This is the content of the famous Hohenberg--Kohn theorem in the standard DFT continuum setting~\cite{hohenberg-kohn1964,Garrigue2018HK}. We will give a comparable result as two separate theorems for our framework.

\begin{theorem}[weak Hohenberg--Kohn result]\label{th:weakHK}
    For two potentials $v_1,v_2\in\V$ that give ground states $\Psi_1,\Psi_2\in\H$ of $H(v_1)$ and $H(v_2)$ respectively, $\densmap(P_{\Psi_1})=\densmap(P_{\Psi_2})$ implies that $\Psi_1$ is also a ground state of $H(v_2)$ and vice versa. The same holds for ensemble states.
\end{theorem}

\begin{proof}
    By the variational principle for ground states it must hold that
    \begin{align}
        &E(v_1) = \langle\Psi_1,H(v_1)\Psi_1\rangle \leq \langle\Psi_2,H(v_1)\Psi_2\rangle,\\
        &E(v_2) = \langle\Psi_2,H(v_2)\Psi_2\rangle \leq \langle\Psi_1,H(v_2)\Psi_1\rangle.
    \end{align}
    Using $\rho = \densmap(P_{\Psi_1})=\densmap(P_{\Psi_2})$ this can be rewritten as
    \begin{equation}\begin{aligned}
        E(v_1) &= \langle\Psi_1,H_0\Psi_1\rangle + \langle\rho,v_1\rangle\leq \langle\Psi_2, H_0\Psi_2\rangle + \langle\rho,v_1+v_2-v_2\rangle= E(v_2) + \langle\rho,v_1-v_2\rangle
    \end{aligned}\end{equation}
    and by switching $1\leftrightarrow 2$ in an analogous way as
    $
        E(v_2) \leq E(v_1) + \langle\rho,v_2-v_1\rangle.
    $
    Combining these two inequalities yields
    \begin{equation}\label{eq:HK-energy-diff}
        E(v_1)-E(v_2)=\langle\rho,v_1-v_2\rangle.
    \end{equation}
    Since $\langle\Psi_1,H(v_2)\Psi_1\rangle = E(v_1)-\langle\rho,v_1-v_2\rangle$, this gives $E(v_2) = \langle\Psi_1,H(v_2)\Psi_1\rangle$, which means that $\Psi_1$ is also a ground state of $H(v_2)$. The same argument can be used to show that $\Psi_2$ is also a ground state of $H(v_1)$. The proof for ensemble states works analogously using $\trace(H(v_i)\Gamma_i)$ for the ground-state energies.
\end{proof}

Apart from providing the proof of half of the usual Hohenberg--Kohn theorem, which we will formulate as the ``strong Hohenberg--Kohn result'' below, the relevance of this theorem is that it allows to map from $\rho\in\Bens$ to the ground-state energy---if $\rho$ is ensemble ground-state $v$-representable. This is achieved by defining the Hohenberg--Kohn functional $\FHK(\rho) = \trace(H_0\Gamma)$~\cite{hohenberg-kohn1964}, where $\Gamma$ is the ground state of $H(v_1)$ for an arbitrary $v_1\in\V$ that leads to $\rho$ in the ground state. Which $v_1$ we choose does not matter, since they all give the same functional value $\FHK(\rho) = E(v_1)-\langle\rho,v_1\rangle = E(v_2)-\langle\rho,v_2\rangle$ as seen from Eq.~\eqref{eq:HK-energy-diff}. But note that this only holds under the assumption of $v$-representability, which is why the constrained-search functionals of Definition~\ref{def:const-search-func}, which avoid this limitation (that is much more restrictive in the continuum setting~\cite{Garrigue2021}), were introduced in the first place. 
Next follows the strong result that really allows to map from ground-state densities $\rho$ to potentials $v$. Note that it is only formulated for pure ground states.

\begin{theorem}[strong Hohenberg--Kohn result]\label{th:strongHK}
    If for $\rho\in\Breg$ we find a $v\in\V$ and $\Psi\in\H$ such that $\Psi$ is a ground state of $H(v)$ and $\Psi\mapsto\rho$ (\emph{pure ground-state $v$-representability}) then this $v$ is unique modulo $\iota^{-1}(\mathbb{R}I_{\H})$.\\
    If additionally $\iota: \V\rightarrow \Lsa(\H)$ is injective and $I_\H\notin\iota(\V)$, then $v$ is unique.
\end{theorem}

\begin{proof}
    Let $v'\in\V$ be a another potential such that $\Psi'\mapsto\rho$ a ground state of $H(v')$. Then by Theorem~\ref{th:weakHK}, we may choose $\Psi'=\Psi$ since the $H(v)$ and $H(v')$ must share a ground state.
    Since $\Psi$ is an eigenstate for both $H(v)$ and $H(v')$, we can subtract the corresponding Schrödinger equations to get 
    \begin{equation}
        \iota(v-v')\Psi = (E-E')\Psi.
    \end{equation}
    Since $P=P_\Psi$ is assumed regular, by  Definition~\ref{def:crit}, it follows that $\iota(v-v')\propto I_\H$, so it must be that $\iota(v-v')=(E-E')I_\H$.
    With the additional assumptions that $\iota: \V\rightarrow \Lsa(\H)$ is injective and $I_\H\notin\iota(\V)$ we must have $v=v'$ and $E=E'$.
\end{proof}

This result, on the same level of generality, was previously stated by
\citet{Xu2022} and for systems with temperature by \citet{Palamara2024}. Other
references often overlook the regularity constraint or additionally assume
that a ground state always connects to a unique potential $v$~\cite{Wu2006},
which means that the weak result directly implies the strong one. 
We note that in case of \emph{pure excited-state $v$-representability}, that is
sometimes necessary in the context of
Theorem~\ref{th:v-rep}\ref{th:v-rep:item:pure} when no representability by pure
ground states can be established, explicit counterexamples to uniqueness of the
potential can be found (see Appendix~\ref{app:2qubit}).

The strong Hohenberg--Kohn result now allows for additional regularity of $\Fens$ on the regular values set $\Breg$ since a unique $-v\in\partial\Fens(\rho)$ means that $\partial\Fens(\rho)$ is a singleton and this in turn implies that $\Fens$ is Gâteaux differentiable~\cite[Th.~25.1]{Rockafellar-book-1970}.

\begin{corollary}
    $\Fens$ is Gâteaux differentiable on $\Breg$.
\end{corollary}

\subsection{Disagreement of pure-state and ensemble functionals}

Since the definitions and properties of the two functionals in Definition~\ref{def:const-search-func} differ, 
it is important to be able to characterize where they agree and where not. 
The following lemma tells that inequality of the functionals must always stem from degeneracy of the ground state.

\begin{lemma}\label{lem:degen}
    If $\rho\in\Bpure$ is pure ground-state $v$-representable, then $\Fpure(\rho)=\Fens(\rho)$. 
    Conversely, for all $\rho\in\relint\Bens$ where $\Fpure(\rho)\neq\Fens(\rho)$, 
    it holds that any representing ground state must be non-pure and thus
    originates from degeneracy. In particular, this holds for all
    $\rho\in\relint\Bens\setminus\Bpure$.
\end{lemma}

\begin{proof}
    The pure ground state $\Psi\mapsto\rho$ minimizes Eq.~\eqref{eq:def-Fpure}. As seen in Eq.~\eqref{eq:E-def}, there can be no $\Gamma\mapsto\rho$ that yields a smaller value $\trace(H_0\Gamma)$ than $\Fpure(\rho)$ because $\Psi$ is already a ground state. So we can choose $\Gamma=P_\Psi$ and have $\Fpure(\rho)=\Fens(\rho)$. 

    We now prove the second claim. 
    First, take $\rho\in(\relint\Bens)\cap\Bpure$.
    Note that  
    $\Fpure(\rho)\neq\Fens(\rho)$ implies that $\rho$ is \emph{not} representable by a pure ground-state. 
    However, by Theorem~\ref{th:v-rep}\ref{th:v-rep:item:ens}, we still have representability by an ensemble ground state, which can only be a non-pure state originating from a degeneracy. 
    If $\rho\in(\relint\Bens)\setminus\Bpure$, representability by a pure ground state cannot hold in the first place.
\end{proof}

An example with numerical evidence for this result is given by a two-qubit example discussed in Appendix~\ref{app:2qubit}. The study of values $\rho\in\V^*$ originating from degenerate ground states in relation to functional theories is an interesting topic on itself~\cite{penz2023geometry}. Here, we want to limit its discussion to the simple observation that for a $v\in\V$ where $H(v)$ has a degenerate ground state, the set of all ground-state values $D=\partial E(v)$ (called a \emph{degeneracy region} in Ref.~\cite{penz2023geometry}) can contain multiple elements and that $\Fens$ is affine on $D$ since $\partial\Fens(\rho)=\{-v\}$ for all $\rho\in D$ by Eq.~\eqref{eq:subdiff-rel}.
This concludes the section on functionals arising from the definition of the scope of a functional theory and we continue with some examples in the next section.

\section{Examples for functional theories}
\label{sec:ex-ft}

Here, as an extension of the examples from Section~\ref{sec:ex-geom}, we will
present the scopes and typical choices for the fixed part of the Hamiltonian
$H_0$ for several physically relevant systems.

\subsection{DFT for fermions on a lattice}
\label{sec:lattice-DFT}

The basic setting for this rather simple physical example was already discussed in Section~\ref{sec:lattice-geom}. 
As before, we consider spinless fermions on a finite lattice with $\H_N = \wedge^N\C^M$. 
Take $\V = \R^M$, $\iota(v) = \sum_{i=1}^M v_i a_i^* a_i$, where $a_i$ ($a_i^*$) are the annihilation (creation) operators acting on the fermionic Fock space, and an arbitrary $H_0\in \Lsa(\H_N)$. The scope $\S=(\R^M,\H_N, \iota, H_0)$ 
gives rise to the usual spinless fermionic DFT on a finite lattice.
The density then turns out to be the per-site occupation number.
The addition of spin as an internal degree-of-freedom is straightforward, and depending on the desired scope it one may consider either spin-summed or spin-resolved quantities~\cite{Penz2024-SpinLattice}. The fixed part $H_0$ of the Hamiltonian can include, for example, a graph Laplacian to emulate a discretization of the continuum together with an electron-electron interaction term~\cite{penz2021-Graph-DFT}. 
Such a functional theory, for example, could be used to study the Hubbard model~\cite{carrascal2015hubbard,Qin2022} (when spin is added) with variable site-dependent external potential, for which the fixed part of the Hamiltonian is given by
\begin{equation}
H_0 = -\sum_{\substack{i\neq j \\ \sigma\in\{\uparrow,\downarrow\}}} t_{ij}a_{i\sigma}^* a_{j\sigma} + U\sum_i a_{i\uparrow}^* a_{i\downarrow}^* a_{i\uparrow} a_{i\downarrow}.
\end{equation}

\subsection{RDMFT for fermions on a lattice}
\label{sec:RDMFT}

In RDMFT, one is concerned with Hamiltonians of the form
$H(h) = h + W$, where $h$ is any one-particle operator and $W$ is a fixed interaction term (e.g., Hubbard on-site)
{\cite{Liebert2023-refining}}.
Let $\H_1$ denote the single-particle Hilbert space, and let
$\H_N$ be the $N$-particle (fermionic or bosonic) Hilbert space.
Consider the mapping to $N$-body operators
\begin{equation}\label{eq:N-body-map}
\begin{aligned}
    \iota_N: \;&\Lsa(\H_1) \rightarrow \Lsa(\H_N)\\
    & h \mapsto \sum_{n=1}^N I^{\otimes(n-1)}\otimes h\otimes I^{\otimes (N-n)}.
\end{aligned}
\end{equation}
The image $\iota_N(\Lsa(\H_1))$ is precisely the space of one-particle operators on $\H_N$.
The corresponding functional theory is then given by the scope
\begin{equation}
    \S_N=(\Lsa(\H_1),\H_N,\iota_N,W).
\end{equation}
The density map $\mu : \densmat(\H_N) \rightarrow \Lsa(\H_1)^*$ is just the partial trace $N\trace_{N-1}(\cdot)$
upon identifying $\Lsa(\H_1)^*$ with $\Lsa(\H_1)$ in the usual way.
Thus, in this case, the density map sends an $N$-particle state to its one-particle reduced density matrix (1RDM).

It is easy to characterize the set $\Bens$ of all $\gamma=\mu(\Gamma)$ with $\Gamma\in\mathcal{E}(\H)$ with the conditions $0\le\gamma\le 1$ and $\trace\gamma=N$. Yet, characterizing the pure-state observable range $\Bpure$ is much more challenging. 
Describing the set of 1RDMs compatible with some pure $N$-particle quantum state is known as the pure-state one-body $N$-representability problem, or in the context of quantum information theory, as (a special case of) the quantum marginal problem. Its solution leads to the generalized Pauli constraints 
\cite{borlandConditionsOnematrixThreebody1972,klyachkoQuantumMarginalProblem2006,klyachkoPauliExclusionPrinciple2009,altunbulakPauliPrincipleRepresentation2008}, which provide necessary and sufficient constraints on the eigenvalues of $\gamma$. The derivation of these constraints is rooted in the Lie-algebra structures discussed in Example~\ref{example:fermion_rdmft}. More on the universal functionals defined from this scope can be found in Ref.~\cite{Liebert2023-refining}.

In quantum chemistry, 
one often works with a fixed set $(\phi_i)_i$ of $L^2(\R^3)$ functions (atomic or molecular orbitals), usually taken to be normalized, linearly independent, but not necessarily orthogonal, that are the building blocks of all considered many-particle wave functions. 
The occupation-number operators $n_i=a_i^*a_i$ are then \emph{non-local} in the sense that they do not correspond to disjoint parts of coordinate space (since the $\phi_i$ in general will have an overlap), so this formulation does not fit to DFT, where the density and potential always have a local character. But in the context of RDMFT, considering self-adjoint one-body operators
spanned by $a_i^*a_j$, no such problem occurs since 1RDMs and the dual elements, one-particle operators themselves, are non-local anyway. This means the RDMFT formulation is equally able to represent lattice models~\cite{LpezSandoval2002} as well as basis sets~\cite{Giesbertz2019}.

\subsection{Spin systems}
\label{sec:spin-chain}

There is a line of work on density-functional theory 
for spin systems
in the literature
\cite{liberoSpindistributionFunctionalsCorrelation2003,capelleSpindensityfunctionalTheoryOpen2005,alcarazDensityFunctionalFormulations2007,prataSpindensityFunctionalExchange2009,maoTestingDensityFunctional2021,zhaoDynamicalExchangecorrelationPotential2023}.
In this setting, the role of density is played by the magnetization at each site, which couples to a site-dependent magnetic field, playing the role of the external potential. Consider the inhomogeneous transverse-field Ising model on a one-dimensional lattice of $m$ sites with open boundary conditions as an example. Generalization to higher spins and other interactions is straightforward.

The family of Hamiltonians is parameterized by an external magnetic field $h \in \mathbb{R}^m$,
\begin{equation}
\label{eq:spin_chain_hamiltonian}
H(h) = -\sum_{i=1}^{m-1}X_iX_{i+1}
+ \sum_{i=1}^m h_i Z_i,
\end{equation}
which act on the Hilbert space $\H=(\mathbb{C}^2)^{\otimes m}$. We choose $H_0 = -\sum_{i=1}^{m-1} X_iX_{i+1}$, $\V=\R^m$ and define the embedding $\iota:\R^m\to\Lsa(\H)$ through $\iota(e_i)=Z_i$ for $i\in\{1,\ldots,m\}$. The scope is thus given by
\begin{equation}
\S = (\R^m,\H,\iota,H_0).
\end{equation}
The density map is then $\mu(\Gamma)=(\braket{Z_1}_\Gamma, \ldots, \braket{Z_m}_\Gamma) \in \V^*\cong\R^m$.

Since the $Z_i$ all commute, the observable ranges $\Bpure, \Bens$ agree and are both given by the convex hull of the simultaneous eigenvalues (see Proposition~\ref{prop:ranges}\ref{prop:ranges:item:commute}), which in this case are exactly the vectors $(b_1, \ldots, b_m)$ for which $b_i = \pm 1$. Hence, we have $\Bpure = \Bens = [-1,1]^m$, the $m$-dimensional hypercube. 

Although the ground state and ground-state energy of the Hamiltonian in Eq.~\eqref{eq:spin_chain_hamiltonian}
can easily be obtained numerically via a Bogoliubov transform \cite{mbengQuantumIsingChain2024}, it serves as a toy model for developing and understanding functional approximations for more general spin-chain functional theories, with analogs of the local density approximation and generalized gradient expansion developed and investigated in the literature~\cite{maoTestingDensityFunctional2021}.

\subsection{Double qubit}
\label{sec:double-qubit-FT}

We note that the previous example about the observable range of a double qubit system with $\H=\C^2\otimes\C^2$ and observables $\{X_1,Y_1,Z_1,Z_2\}$ from Section~\ref{sec:geometry_double_qubit} is extended in Appendix~\ref{app:2qubit} to include also the calculated functional values at certain points. In particular, this example shows that the ensemble constrained-search functional is strictly smaller than the pure-state constrained-search functional away from the boundary.

\section{Change of scope}
\label{sec:change-of-scope}

Here, we discuss several constructions that modify the scope of a functional theory.
As we will see, the constrained-search functionals associated to the new scope can often be related to the original functionals in simple ways,
allowing us to transfer properties of functionals from one scope to another.

\subsection{Reduction of scope}
\label{sec:reduction-of-scope}

Given a scope $\S = (\V, \H, \iota, H_0)$,
a linear map $f: \V_\mathrm{r} \rightarrow \V$, 
and a fixed potential $v_0 \in \V$, 
we can consider the \textit{reduced scope} 
\begin{equation}
    \S_\mathrm{r} := (\V_\mathrm{r}, \H, \iota\circ f, H_0 + \iota(v_0)).
\end{equation}
Let $f^*: \V^* \rightarrow \V_\mathrm{r}^*$ be the dual of $f$. 
Then, the density map of the reduced functional theory is given by
\begin{equation}
\mu_\mathrm{r} = f^*\circ\mu,
\end{equation}
where $\mu: \densmat(\H) \rightarrow \V^*$ is the density map for $\S$.
This then implies the following
.\begin{proposition}
\label{prop:scope_reduction}
The pure-state observable range of the reduced scope is $\Bpure^{(\S_\mathrm{r})} = f^*(\Bpure^{(\S)})$.
The pure-state constrained-search functional $\Fpure^{(\S_\mathrm{r})}$
is given by
\begin{equation}
\label{eq:scope_reduction}
\Fpure^{(\S_\mathrm{r})}(\rho')
 = \min_{\rho \in (f^*)^{-1}(\rho')} \left\{ \braket{\rho, v_0} + \Fpure^{(\S)}(\rho) \right\}.
 \end{equation}
Analogous results hold for $\Bens^{(\S_\mathrm{r})}$ and $\Fens^{(\S_\mathrm{r})}$.
\begin{proof}
    \begin{equation}
    \begin{aligned}
    \Fpure^{(\S_r)} (\rho')
    &:= \min_{P \in (f^*\circ \mu)|_{\puremf(\H)}^{-1}(\rho')} \trace(P(H_0 + \iota(v_0))\\
    &\;= \min_{\rho \in (f^*)^{-1}(\rho')} \min_{P\in \mu|_{\puremf(\H)}^{-1}(\rho)} 
    \trace(P(H_0 + \iota(v_0))\\
    &\;= \min_{\rho \in (f^*)^{-1}(\rho')} \left\{ \braket{\rho, v_0} + \Fpure^{(\S)}(\rho) \right\}.
    \end{aligned}
    \end{equation}
\end{proof}
\end{proposition}
That is, the value of the new functional at $\rho'\in \V_\mathrm{r}^*$
is nothing but value of the old one minimized over the fiber $(f^*)^{-1}(\rho')$. 

\begin{example}
    For any Hilbert space $\H$, we can construct the
    \textit{tautological scope}, defined as
    \begin{equation}
        \S_{\text{taut}} := (\Lsa(\H), \H, \mathrm{Id}_{\Lsa(\H)}, 0).
    \end{equation} 
    Note that $\Fpure^{(\S_{\text{taut}})}$ identically vanishes.
    Let $\iota: \V \rightarrow \Lsa(\H)$ be any linear map, and let $H_0\in \Lsa(\H)$ be any self-adjoint operator.
    In this case, Eq.~\eqref{eq:scope_reduction} 
    simply becomes the definition of the constrained-search functional for the scope $(\V, \H, \iota, H_0)$.
\end{example}

\begin{example}
    Take $\mathcal{H} = \wedge^N\mathbb{C}^d$ to be the Hilbert space of 
    $N$ spinless fermions on $d$ lattice sites. 
    Let $(\Lsa(\mathbb{C}^d), \H, \iota, H_0)$
    be the RDMFT scope, where $\iota(v) = \sum_{ij=1}^dv_{ij}a_i^* a_j$. 
    Let $f: \mathbb{R}^d\rightarrow \Lsa(\mathbb{C}^d)$
    be the map that sends $(v_1, \dotsb, v_d)$ to $\mathrm{diag}(v_1, \dotsb, v_d)$,
    and let $t\in \Lsa(\mathbb{C}^d)$ be some fixed kinetic operator (e.g. one consisting of hopping terms). Then the reduction with respect to $f$ and $t$
    is the reduction from RDMFT to (lattice) DFT.
\end{example}

\subsection{Purification of ensemble functionals}

\begin{theorem}
\label{thm:ensemble_is_pure}
Let $\S = (\V, \H, \iota, H_0)$ be any scope, 
and let $\Fens^{(\S)}$ be the corresponding ensemble functional. 
Consider the scope $\S' = (\V, \H\otimes \H, \iota', H_0\otimes I_\H)$, where
$\iota'(v) := \iota(v)\otimes I_\H$, and
let $\Fpure^{(\S')}$ denote its pure functional. Then
$\Fens^{(\S)} = \Fpure^{(\S')}$.
\end{theorem}

\begin{proof}
First we show $\dom\Fens^{(\S)}= \dom\Fpure^{(\S')}$. 
In other words, we want to show that $\Bens^{(\S)} = \Bpure^{(\S')}$, where $\Bpure^{(\S')}$ is the pure-state observable range of the functional theory $\S'$. Any ensemble state $\Gamma\in \densmat(\H)$ is the partial trace of a pure state on $\H\otimes \H$, i.e., $\Gamma = \trace_2(P_{\tilde \Psi})$ for some normalized $\tilde \Psi \in \H\otimes \H$ (this is the usual purification of quantum states~\cite[Sec.~2.5]{nielsen-chuang-book}). So $\trace(\iota(v)\Gamma) = \braket{\tilde\Psi,  (\iota(v)\otimes I)\tilde\Psi}$, which shows $\Bens^{(\S)}\subseteq \Bpure^{(\S')}$. By the same argument backwards, we have $\Bpure^{(\S')} \subseteq \Bens^{(\S)}$.
Take any $\rho \in \Bens^{(\S)} = \Bpure^{(\S')}$. Then
\begin{equation}
\begin{aligned}
\Fens^{(\S)}(\rho) &= \min\{\trace(\Gamma H_0) : \mu(\Gamma) = \rho\}\\
&=\min \{\braket{\tilde\Psi, (H_0\otimes I_\H)\tilde\Psi} :\tilde \Psi \in \H\otimes \H, \forall v\in \V: \braket{\tilde\Psi, \iota(v)\otimes \tilde\Psi} = \braket{\rho, v} \}= \Fpure^{(\S')}(\rho),
\end{aligned}
\end{equation}
where the second equality again follows from the purification of ensemble states.
\end{proof}

\subsection{\texorpdfstring{Purification of $\bm{w}$-ensemble functionals}{Purification of w-ensemble functionals}}


Finally, we turn to the so-called $w$-ensemble functional theories \cite{grossDensityfunctionalTheoryEnsembles1988,SP21,LCLS21,Gould2026-Ensemblization}, which allow one to target excited states through ensembles with prescribed spectral weights. We show that they are naturally encompassed by our framework through a purification construction. To this end, let $\S=(\V,\H,\iota,H_0)$ be a scope with $\H\cong\C^d$, and let $w=(w_1,\ldots,w_d)\in\mathbb{R}^d$ be a nonincreasing sequence of auxiliary weights satisfying $\sum_{i=1}^{d}w_i=1$ and $w_i\ge0$.

\begin{definition}
The $w$-ensemble functional is
\begin{equation}\begin{aligned}
\mathcal{F}_w(\rho)
 := \min \{\trace(\Gamma H_0) :\; 
 &\Gamma \in \densmat(\H),\; \densmap(\Gamma) = \rho,\;
 \specdown(\Gamma) = w
 \},
\end{aligned}\end{equation}
where $\specdown(\Gamma)$ denotes the spectrum of $\Gamma$ ordered nonincreasingly.
\end{definition}

The utility of the $w$-ensemble constrained-search functional is the following.
\begin{proposition}
Let $v\in \V$ be any potential, 
and let $E_1\le E_2 \le \ldots \le E_d$ be the eigenvalues 
of the Hamitonian $H(v) = \iota(v) + H_0$, ordered nondecreasingly.
Then
\begin{equation}
\sum_{i=1}^d w_i E_i =  
\min_{\rho} \left\{\braket{\rho, v} + \Fw(\rho)\right\}.
\end{equation}
\end{proposition}

The proof is a simple application of the $w$-ensemble variational principle
\cite{grossDensityfunctionalTheoryEnsembles1988,dingGroundExcitedStates2024}.
We will show that $\Fw$ can be thought of as a `slice' of the pure functional of a suitably enlarged scope.

\begin{theorem}
Consider the enlarged scope
\begin{equation}
 \S' = 
 \left(\V\oplus \Lsa(\H), \H\otimes \H, \iota', H_0\otimes I_{\H}\right),
\end{equation}
where $\iota'(v, A) = \iota(v)\otimes I_{\H} + I_{\H}\otimes A $.
We can canonically identify
\begin{equation}
(\V\oplus\Lsa(\H))^* \cong \V^* \oplus \Lsa(\H)^* \cong \V^* \oplus \Lsa(\H).
\end{equation}
That is,
a density $\tau$ in 
the functional theory defined by the scope $\S'$ can be written as
$\tau = (\rho, \nu)$, where $\rho \in \V^*$
and $\nu\in \Lsa(\H)$. 
Then for any ensemble state $\nu\in \densmat(\H)$ it holds with $w=\specdown(\nu)$ that
\begin{equation}
\mathcal{F}_{w}^{(\S)}(\rho)
 = \Fpure^{(\S')}(\rho, \nu).
\end{equation}
\end{theorem}

\begin{proof}
\begin{equation}
\begin{aligned}
\Fpure^{(\S')}(\rho, \nu)
&= 
\min\Big\{\braket{\Psi,(H_0\otimes I)\Psi}
 : \Psi \in \H\otimes \H,\mu(\trace_2(P_\Psi)) = \rho, \trace_1(P_\Psi) = \nu
\Big\}\\
&= 
\min\Big\{\braket{\Psi,(H_0\otimes I)\Psi}
: \Psi \in \H\otimes \H,\mu(\trace_2(P_\Psi)) = \rho, \specdown(\trace_2(P_\Psi)) = \specdown(\nu)
\Big\}\\
&=\min\Big\{\trace(\Gamma H_0)
: \Gamma \in \densmat(\H),\densmap(\Gamma) = \rho,\specdown(\Gamma) = \specdown(\nu)
\Big\} = \Fw^{(\S)}(\rho),
\end{aligned}
\end{equation}
where in the second equality we used the fact that the two partial traces of a pure state in a bipartite
system share the same spectrum.
\end{proof}

This result, like Theorem~\ref{thm:ensemble_is_pure}, is very useful, since it shows that statements about the pure functional could be transferred to the $w$-ensemble functional.

\section{Symplectic-geometric perspective}
\label{sec:Lie}

The general notion of a functional theory introduced in Section~\ref{sec:def} is formulated at an abstract level in terms of the tuple $(\mathcal{V},\H,\iota,H_0)$, the density map $\mu$ of Eq.~\eqref{eq:density-map}, and the functionals of Eqs.~\eqref{eq:def-Fpure}--\eqref{eq:E-def}. In many cases of interest, however, the observable space $\iota(\V)$ carries additional Lie-algebraic structure. In this section, we consider the case in which this structure is induced by a unitary representation of a Lie group $K$. Then the density map admits a moment map interpretation, and the abstract map $\iota$ is the self-adjoint form of the corresponding infinitesimal Lie algebra representation. The main point is that different functional theories arise from different choices of $K$ and its representation, while the resulting reduced variables, constrained-search formulations, and representability questions are governed by the same moment map framework. Section~\ref{sec:Lie-gen} introduces the general formalism, Section~\ref{sec:Lie-consequences} discusses its consequences for functional theories, and Section~\ref{sec:Lie-examples} applies the framework to lattice DFT, RDMFT, and symmetry-adapted variants relevant to electronic structure theory.

The identification of the density map with the moment map is important because
it allows one to use the standard results from the theory of moment maps, such
as equivariance, fiber structure, and convexity theorems, in the analysis of
functional theories (see Section~\ref{sec:Lie-consequences} for more details).

\subsection{General formalism \label{sec:Lie-gen}}

As in the previous sections, we work with a finite-dimensional Hilbert space $\H$. 
Recall that the tangent space to the manifold $\puremf(\H)$ of pure states 
at $P\in \puremf(\H)$ is given by $T_P \puremf(\H) = \{[S, P] : S \in \rmi \Lsa\}$. 
The \textit{Fubini--Study symplectic form} $\omega^{\text{FS}}$ on $\puremf(\H)$ is defined as~\cite{Kibble79, AS99, BH01}
\begin{equation}
    \omega_P^{\text{FS}} ([S, P], [T, P] ) =
    \rmi\trace(P[S,T]).
\end{equation}
This is well-defined: if we have $[S,P] = [S',P]$, then $\trace(P[S',T]) = \trace([P,S']T) = \trace([P,S]T) = \trace(P[S,T])$. The Fubini--Study symplectic form is relevant because it turns the group action of a Lie group $K$ on pure states into a Hamiltonian group action, whose moment map computes precisely the reduced variable defined by the expectation values of the chosen observables, as will be discussed next.

With respect to the Fubini--Study symplectic form, the natural $K$-action on $\puremf(\H)$ 
admits a moment map $\mu_0: M\rightarrow \mathfrak{k}^*, P \mapsto (S\mapsto \rmi \trace(PS))$, where $\mathfrak{k}$ is the Lie algebra of the Lie group $K$. For example, for $K=\mathrm{U}(\H)$ it consists of skew-adjoint operators on $\H$ (see Appendix~\ref{app:symplectic} for definitions). Since the moment map encodes the group action in terms of expectation value functions, it relates directly to the density map introduced in Eq.~\eqref{eq:density-map}. The moment map adds additional structure to functional theories, which we will explain in Section~\ref{sec:Lie-consequences}.

To connect this to the general framework introduced in Section~\ref{sec:def}, take 
a Lie subgroup $K\subseteq \mathrm{U}(\mathcal{H})$, and set $\V := \rmi \mathfrak{k}$,
where $\mathfrak{k}\subseteq \mathfrak{u}(\mathcal{H})$ is the Lie algebra of $K$. 
Recall that the density map $\mu: \puremf(\H)\rightarrow \V^*$ is defined by $\braket{\mu(P), v} = \trace(P\iota(v))$ for $P\in \mathbbm{P}(\H)$ and all $v\in\V$.
We can identify the space of densities $\V^*$ with the dual of the Lie algebra $\mathfrak{k}^*$ via the isomorphism $\V^* = (\rmi \mathfrak{k})^* \rightarrow \mathfrak{k}^*, \rho \mapsto \rmi \rho$.
The following is then a direct consequence of Lemma~\ref{lem:mom_map_composition} in Appendix~\ref{app:symplectic} (see also Ref.~\cite{rosensteelNonAbelianDensityFunctional1998}).

\begin{theorem}
\label{thm:densmapismommap}
The density map $\mu: \puremf(\H)\rightarrow \V^* = (\rmi \mathfrak{k})^* \cong \mathfrak{k}^*$ is a moment map for the $K$-action
on $\puremf(\H)$.
\end{theorem}

In other words, whenever the space of observables forms a Lie algebra (up to $\rmi$), the pure-state density map is a moment map (up to identifying $(\rmi\mathfrak{k})^* \cong \mathfrak{k}^*$). Consequently, the Lie-algebraic moment map perspective provides a unified framework: structural results can be proved once for a Hamiltonian $K$-action and then specialized to a particular functional theory by choosing the group $K$, with Theorem~\ref{thm:densmapismommap} identifying the resulting moment map with the corresponding pure-state density map, and hence with the reduced variable of that theory. This will be illustrated for the chemically relevant settings of lattice DFT, RDMFT and symmetry-adapted variants in Section~\ref{sec:Lie-examples}. 

Let $K$ be a Lie group and $b:K\to\mathrm{U}(\H)$ a Lie group homomorphism,
with $b_*: \mathfrak{k} \rightarrow \mathfrak{u}(\H)$
denoting the corresponding Lie algebra representation. 
The representation $b_*$ specifies how these generators act on a general Hilbert space $\H$. In applications, different choices of $b_*$ allow the same Lie algebra to act, for example, on a one-particle Hilbert space $\H_1$ or on the full $N$-particle Hilbert space $\H$.
Then, the map $\iota$ in Definition~\ref{def:scope} is nothing else than a self-adjoint variant of the infinitesimal Lie algebra representation $b_*$,
\begin{equation}\label{eq:iota-Lie}
    \iota:\rmi \mathfrak k\to\rmi\mathfrak u(\H)\cong\mathcal{L}_{\rm sa}(\H),\quad \iota(v)=\mathrm{i} b_*(-\rmi v),
\end{equation}
which further clarifies the meaning of the map $\iota$ in the context of functional theories derived from a Lie group $K$. 
In the electronic structure examples in Section~\ref{sec:Lie-examples}, $\mathfrak{k}$ is typically the Lie algebra of a one-particle unitary group, or of a subgroup thereof, and the representation $b$ is the induced action on the $N$-particle Hilbert space. Hence, $\iota$ has a direct physical interpretation: it assigns to each chosen one-particle or symmetry generator the corresponding self-adjoint many-body observable whose expectation value enters the reduced variable of the functional theory.
Together with an arbitrary but fixed $H_0\in \mathcal{L}_{\mathrm{sa}}(\mathcal{H})$, Eq.~\eqref{eq:iota-Lie} defines a scope $(\rmi \mathfrak{k}, \H, \iota, H_0)$.
For this scope, the corresponding density map is given by 
\begin{equation}
    \langle\mu_K(\Gamma),v\rangle = \trace(\Gamma\iota(v)),\quad v\in\rmi\mathfrak k.
\end{equation}
Its restriction to pure states, $\mu_K:\mathbbm{P}(\H)\to(\rm i\mathfrak k)^*$, is a moment map for the $K$-action on $\mathbbm{P}(\H)$, as described by Theorem~\ref{thm:densmapismommap}. On $\densmat(\H)$, the same
formula gives the affine extension of this moment map.

Furthermore, in this case, there is a convenient geometric description of the kernel of the derivative of the density map in terms of the symplectic structure and the $K$-action. Namely, for any pure state $P\in \puremf(\H)$, it holds that~\cite{CdS01}
\begin{equation}
\label{eq:Lie-fiber-tangent}
\ker(\rmd _P\mu) = 
(T_P (K\cdot P))^{\omega^{\mathrm{FS}}},
\end{equation}
where $K\cdot P$ denotes the $K$-orbit through $P$ and $(\cdot)^{\omega^{\mathrm{FS}} }$ denotes the symplectic complement. 
If $\rho \in \Bpure = \mu(\puremf(\H))$ is a regular value, then Eq.~\eqref{eq:Lie-fiber-tangent}
gives a geometric characterization of the tangent space to the constrained manifold $\mu^{-1}(\rho)$,
which could be useful for carrying out the pure-state constrained  search.

\subsection{Consequences for functional theories \label{sec:Lie-consequences}}

The fact that the density map $\mu: \puremf(\H)\rightarrow \V^*$ is a moment
map adds several structural ingredients that are useful both conceptually and
practically. 
First, the reduced variable is no longer just an arbitrary collection of expectation
values: it is a natural geometric object attached to the action of $K$ on the pure-state manifold. In particular, the moment map is equivariant (see Lemma~\ref{lem:moment-map-equivariance}),
\begin{equation}\label{eq:Lie-equivariance}
\mu(k\cdot P)=\Ad_k^*\mu(P),\quad k\in K,\quad P\in \puremf(\H),
\end{equation}
where $\operatorname{Ad}_k^*$ denotes the coadjoint action of $K$ on $\V^*= \rmi \mathfrak k^*$. 
An important implication of equivariance is that pure state representability of the 
reduced variable $\rho\in\V^*$ is invariant under the coadjoint action. 
Indeed, if $\rho=\mu(P)$ for some $P\in\puremf(\H)$, then $\Ad_k^*\rho=\mu(k\cdot P)$ is again pure-state representable for every $k\in K$. If, in addition, the universal part $H_0$ is $K$-invariant, then the constrained-search functional is constant along these orbits as well. 
In RDMFT, for example, the first statement 
corresponds to the basic fact that the $N$-representability of a 1RDM depends only on its spectrum (see Example~\ref{example:fermion_rdmft} below), 
even though the universal functionals themselves need not be invariant under orbital rotations, unless
more symmetry is present.

More importantly, due to Theorem~\ref{thm:densmapismommap}, well-developed tools from symplectic geometry are directly applicable to the $N$-representability problem. 
Examples of such results are the convexity theorems (Theorem~\ref{thm:atiyah}
and Theorem~\ref{thm:kirwan} in Appendix~\ref{app:symplectic}).  
In the abelian case, Theorem~\ref{thm:atiyah} implies that $\mu(\puremf(\H)) = \Bpure$, the domain of the pure-state constrained-search functional, is a convex polytope, which is consistent with
Proposition~\ref{prop:ranges}\ref{prop:ranges:item:commute}. 
For a general compact connected Lie subgroup $K\subseteq \mathrm{U}(\H)$, the pure-state observable range $\Bpure\subseteq \V^*$ need not be convex. 
However, after choosing a Cartan subalgebra $\mathfrak{t}\subseteq \mathfrak{k}$ and a positive Weyl chamber $\mathfrak{t}_+^*\subseteq\mathfrak{t}^*$, the intersection
\begin{equation}\label{eq:Lie-moment-polytope}
\Delta:=\mu(\puremf(\H))\cap\mathfrak t_+^*
\end{equation}
\textit{is} a convex polytope by Kirwan's theorem (Theorem~\ref{thm:kirwan}). 

\begin{example}
\label{example:fermion_rdmft}
In fermionic RDMFT (see also Section~\ref{sec:Lie-examples}),
the facet-defining inequalities of $\Delta$ are the so-called
\textit{generalized Pauli constraints} and their symmetry-adapted analogues \cite{klyachkoQuantumMarginalProblem2006,AK08, klyachkoPauliExclusionPrinciple2009,SGC13,Schilling2014,LLAMOS25}.
\end{example}

Finally, we remark that even if the space of potentials is not a Lie algebra,
the techniques introduced in this section could still be relevant. To illustrate this, let
$\S = (\V, \H, \iota, H_0)$ with $\V=\rmi\mathfrak{k}$ be a scope, where
$\mathfrak{k}\subseteq\mathfrak{u}(\mathcal{H})$ is a Lie subalgebra as above, and suppose
$\V_{\rm r}\subseteq \V$ is any linear subspace. 
As discussed in Section~\ref{sec:reduction-of-scope}, the density map for $\V_{\rm r}$
is simply $\mu_{\rm r}=f^* \circ \mu$, where $f^*: \V^* \rightarrow \V_{\rm r}^*$
is the dual map of the inclusion and $\mu$ is the moment map for the $K$-action.
Because $\V_{\rm r}$ is not necessarily a Lie algebra, we cannot always interpret $\mu_{\rm r}$
as a moment map. Consequently, the standard equivariance and convexity theorems apply  only to $\mu$, not to $\mu_{\rm r}$.
Nonetheless, even in that situation, the finer moment map geometry still provides useful indirect information because the observable range for the projected variable is
$\mathcal{B}_p^{(\S_{\rm r})}=\mu_{\rm r}(\puremf(\H))=f^*(\Bpure^{(\S)})$.
Hence, every exact description of $\Bpure^{(\S)}$ and, more generally, every outer approximation $C\supseteq \Bpure^{(\S)}$, immediately yields necessary representability constraints for the projected variable through $\Bpure^{(\S_{\rm r})}\subseteq f^*(C)$ following Proposition~\ref{prop:scope_reduction}.

\subsection{Examples for Lie-algebra functional theories}
\label{sec:Lie-examples}

The general formalism introduced in Section~\ref{sec:Lie-gen} specializes to several Lie algebra settings that already appear in Sections~\ref{sec:ex-geom} and \ref{sec:ex-ft}. 
The three cases below illustrate the mechanisms that are most relevant in electronic structure theory: an abelian subgroup, the full one-body moment map, and spin-symmetry adapted settings. 

\subsubsection{Abelian subgroup: Lattice DFT\label{sec:symplectic:lattice-dft}}

Consider $N$ spinless fermions on $d$ lattice sites with Hilbert space $\H_N = \wedge^N \H_1$, 
where $\H_1 = \mathbb{C}^d$ is the one-particle Hilbert space.
Denote by $n_i=a_i^*a_i$ the number operator at site $i\in\{1,2,\ldots, d\}$.
The action of a group element $g\in \mathrm{U}(\H_1)$ on the anti-symmetric $N$-particle Hilbert space $\H$ is described by the exterior-power representation
\begin{equation}\label{eq:ex-power-rep}
    b:\mathrm U(\H_1)\to \mathrm U(\H_N),\quad b(g)=\wedge^N g =\left. g^{\otimes N}\right|_{\wedge^N\H_1},
\end{equation}
with the corresponding Lie algebra homomorphism $b_*: \mathfrak{u}(\H_1) \rightarrow \mathfrak{u}(\H_N)$. 

As we will see next, lattice DFT emerges from the general Lie algebraic structure provided by the moment map by choosing the subgroup $K$ to be certain maximal torus. 
To be more precise, let $T = \mathrm{U}(1)^d\subseteq \mathrm{U}(d) = \mathrm{U}(\H_1)$ be the maximal torus of diagonal unitary matrices, and let  $j:T\to \mathrm{U}(\H_1)$ denote the inclusion. We consider the scope $(\rmi\mathfrak t,\H_N, \iota,H_0)$ with $\iota$ being the self-adjoint variant of $b_*\circ j_*$ as in Eq.~\eqref{eq:iota-Lie}. 
Then, after composing the map $\iota$ with the mapping to $N$-body operators from Eq.~\eqref{eq:N-body-map}, a real vector $v=(v_1,\dotsc,v_d)$ is mapped to the lattice potential
\begin{equation}
\iota(v)=\sum_{i=1}^d v_i n_i.
\end{equation}
With this notation, the induced $T$-action on $\puremf(\H_N)$ admits a moment map obtained by pairing states with the observables $\iota(v)$. Thus,
\begin{equation}
 \mu_T(P)= \bigl( \trace(Pn_1),\dotsc,\trace(Pn_d)  \bigr), \quad P\in\puremf(\H_N).
\end{equation}
Since $\sum_{i=1}^d n_i=N I_{\H_N}$ we have $\sum_{i=1}^d \trace(Pn_i)=N$. 
Thus, $\mu_T(P)$ is exactly the per-site lattice density, more precisely the vector of site-occupation numbers, already discussed in Secs.~\ref{sec:lattice-geom}. Since the operators $n_1,\dotsc,n_d$ commute, Proposition~\ref{prop:ranges}\ref{prop:ranges:item:commute} implies $\Bpure=\Bens$, and the observable range is the convex hull of the occupation vectors of the site basis. 

\subsubsection{Full one-body unitary group: RDMFT\label{sec:LieRDMFT}}

We now show how choosing a different Lie group for $K$ than the maximal torus leads to RDMFT instead of lattice DFT. This already illustrates one of the key messages of this section, namely that different functional theories can be recovered from their same parent symplectic-geometric structure simply by choosing different Lie groups that enter the moment map in Theorem~\ref{thm:densmapismommap}. Let $\H_1\cong\CC^d$ and let $\H_N =\wedge^N\H_1$ be the corresponding $N$-particle Hilbert space as in Section~\ref{sec:symplectic:lattice-dft}.

Take $K=\mathrm{U}(\H_1)$, which acts on the $N$-particle Hilbert space $\H_N$ through the exterior-power representation $b$ in Eq.~\eqref{eq:ex-power-rep}. With the same self-adjoint convention relating $\iota$ to $b_*$ as in Eq.~\eqref{eq:iota-Lie}, we consider the scope $(\rmi \mathfrak{u}(\H_1),\H_N, \iota,H_0)$. Thus, the two scopes of RDMFT and lattice DFT only differ in the (self-adjoint) Lie algebra $\rmi \mathfrak{u}(\H_1)$.
The associated moment map
\begin{equation}
    \mu_{\mathrm{U}(\H_1)}:\puremf(\H_N)\to \left(\rm\rmi\mathfrak{u}(\H_1)\right)^*
\end{equation}
is characterized by
\begin{equation}
    \langle \mu_{\mathrm{U}(\H_1)}(P),A\rangle = \trace_{\H}\bigl(P\iota(A)\bigr)  = \trace_{\H_1}(\gamma_P A),
\end{equation}
where $P\in\puremf(\H)$, $A\in\Lsa(\H_1)= \rmi\mathfrak{u}(\H_1)$ and $\gamma_P=N\trace_{N-1}(P)$ is the full 1RDM of $P$. Thus, by the trace pairing, the value of the moment map at a pure state $P\in\puremf(\H_N)$ is exactly its 1RDM $\gamma_P$.

This is the Lie algebra formulation of the RDMFT setting discussed in Section~\ref{sec:RDMFT}. Passing to ordered eigenvalues, equivalently to the intersection of the moment map image with a positive Weyl chamber, the pure-state spectral $N$-representability problem for the 1RDM is therefore equivalent to determining the moment (Kirwan) polytope for the $\mathrm{U}(\H_1)$-action on $\mathbbm{P}(\H_N)$~\cite{AK08} (see also Theorem~\ref{thm:kirwan}). Equivalently, one may use the $\mathrm{SU}(\H_1)$-action, in which case the moment map is shifted by the fixed trace part. Its defining linear inequalities consist of the Pauli exclusion principle, the trace condition, and the additional generalized Pauli constraints (see Example~\ref{example:fermion_rdmft}).

\subsubsection{Spin-adapted and product group variants \label{sec:symplectic:spin}}

In the following, we consider identical particles with internal degrees of freedom (e.g., spin), and restrict the discussion to spin-$\tfrac{1}{2}$ fermions. 
The one-particle Hilbert space is $\H_1 = \H_\ell\otimes \H_s\cong \CC^d\otimes \CC^2$,
where we separate the spatial/orbital ($\H_\ell$) and the internal/spin ($\H_s$) degrees of freedom. The two spin degrees of freedom are denoted by $\sigma\in\{ \uparrow, \downarrow \}$. 
We consider two exemplary cases: first, a spin-free orbital adaptation of RDMFT in the case of spin-independent (i.e., $\mathrm{SU}(2)$-invariant) Hamiltonians and, second, a general spin-adapted RDMFT framework that allows for additional terms in the one-particle Hamiltonian that break the $\mathrm{SU}(2)$ symmetry.

In the spin-independent case, both the fixed part $H_0$ and all elements of $\V$ commute with the three generators of the Lie group $\mathrm{SU}(2)$. This implies that the corresponding group $K$ to choose for the moment map is the unitary group $\mathrm{U}(\mathcal{H}_\ell)\cong \mathrm{U}(d)$ whose generators are~\cite{Paldus74, LCLMS24}
\begin{equation}
    E_{ij} = a_{i\uparrow}^* a_{j\uparrow} + a_{i\downarrow}^* a_{j\downarrow}.
\end{equation}
Thus, the scope of the functional theory in the spin-independent setting is the tuple $(\mathrm{i}\mathfrak{u}(\mathcal{H}_\ell), \H_N, \iota, H_0)$ with $\iota: \rmi\mathfrak{u}(\mathcal{H}_\ell)\to \rmi \mathfrak{u}(\H_N)$. The corresponding moment map that yields the reduced variable is 
\begin{equation}\label{eq:mu-SU2}
    \mu_{\mathrm{U}(\H_\ell)}: \mathbbm{P}(\mathcal{H}_N) \to (\rmi\mathfrak{u}(\H_\ell))^*.
\end{equation}
Thus, the reduced variable of the functional theory is no longer the full 1RDM, but the orbital 1RDM 
\begin{equation}\label{eq:gl}
    (\gamma_\ell)_{ij} = \trace(P E_{ji}),\qquad P\in \mathbbm{P}(\H_N).
\end{equation}
Here, we restrict our discussion of the scope as defined in Definition~\ref{def:scope} and the moment map in Eq.~\eqref{eq:mu-SU2}. For more details on spin-adaptation of one-particle reduced density matrix functional theories we refer to Refs.~\cite{LCLMS24, LMS24-jctc, LMS24-prl, Liebert25-phdthesis}. These references also go beyond the present setting by discussing the consequences of further restricting the $N$-particle Hilbert space $\H_N$ to irreducible representations (i.e., symmetry sectors) of $\mathrm{SU}(2)$.

As the second spin-adapted setting, we consider the product group $K=\mathrm{U}(\H_\ell)\times \mathrm{U}(\H_s)$, where `spin-adapted' means that the spin degrees of freedom of the fermions are incorporated explicitly into the construction of the functional theory through the corresponding Lie group structure. 
The Lie algebra of $K$ is then $\mathfrak{k}= \mathfrak{u}(\H_\ell)\oplus \mathfrak{u}(\H_s)$, leading to the scope $(\rmi (\mathfrak{u}(\H_\ell)\oplus \mathfrak{u}(\H_s)), \H_N, \iota, H_0)$. 
Thus, the corresponding potentials $\iota(v), v\in \rmi (\mathfrak{u}(\H_\ell)\oplus \mathfrak{u}(\H_s))$ may break the $\mathrm{SU}(2)$-invariance, for instance, through an external magnetic field. More details on this setting can be found in Refs.~\cite{LMS24-jctc, Liebert25-phdthesis}. 
From the conceptual approach taken in this paper, the product group example is interesting as it leads to a scope whose observable space of variable parts of the Hamiltonian $\iota(\V)$ is described by a direct sum of two Lie algebras. 
Then, the moment map $\mu_{\mathrm{U}(\H_\ell)\times \mathrm{U}(\H_s)}$ yields the tuple $(\gamma_\ell, \gamma_s)$ as the reduced variable. Here, $\gamma_\ell$ is the orbital 1RDM defined in Eq.~\eqref{eq:gl}, and $\gamma_s$ the spin 1RDM, whose matrix elements follow from the generators
\begin{equation}
 S_{\sigma\sigma'} = \sum_{i=1}^d a_{i\sigma}^*a_{i\sigma'}   
\end{equation}
of $\mathrm{U}(2)$ according to
\begin{equation}
    (\gamma_s)_{\sigma\sigma'} = \trace(P S_{\sigma'\sigma}),\qquad P\in \mathbbm{P}(\H_N).
\end{equation}

Both examples, the spin-independent ($K=\mathrm{U}(\H_\ell)$) and the more general spin-adapted setting ($K=\mathrm{U}(\H_\ell)\times \mathrm{U}(\H_s)$) can be generalized to spin-$m$ particles, $m\in\{0, \tfrac{1}{2}, 1, \ldots\}$. This further illustrates the flexibility of the symplectic-geometric approach for incorporating internal degrees of freedom into a functional theory.

\section{Concluding remarks and outlook}
\label{sec:outlook}

There are some obvious directions for possible extensions and further research. One is the obvious extension to infinite-dimensional (separable) Hilbert spaces, which arguably constitutes the typical setting for a formulation of quantum mechanics. This can come in two different forms. 
First, one could still consider a finite-dimensional parameterization of the observable space by $\V$, perhaps with the restriction to observables with discrete spectra. As an intermediate step, this would allow to retain many of the results achieved here, while still substantially extending the ranging of validity. 
The next step would then be to introduce families of observables parameterized by infinite-dimensional vector spaces, such as the density operator $\hat\rho(\rr)$ of DFT.

The other proposed direction is the obvious extension to the time domain, with the principle goal being solving time-dependent Schrödinger equation and obtaining the time-dependent reduced variables (densities) instead. For lattice systems, this problem was pioneered by \citet{tokatly2011time} and \citet{farzanehpour2012time}. More recently, Refs.~\cite{Cances2026-TDDFT-geom-lattice, CDGLLT26} developed a geometric framework for Schr\"odinger dynamics under time-dependent density constraints. 
In the finite-lattice TDDFT setting, this framework yields, besides the usual dynamics derived from stationarity of the action, an alternative constrained dynamics in which the density constraint can be imposed by an imaginary potential, or equivalently by a nonlocal Hermitian operator.
 Compared to ground state DFT, the time-dependent variant is much less studied and misses the huge variety of well-engineered functional approximations for ground states~\cite{Toulouse2022-chapter}. Moreover, even less is known in the case of non-abelian functional theories such as RDMFT \cite{PG16}. The field would thus definitely benefit from a mathematically sound framework, even if it is only for the finite-dimensional domain.

Finally, we note that the present manuscript focuses on an \emph{irreducible} formulation of functional theories. Here, irreducible means that the reduced variable is an element of the range of the density map. In this sense, it contains exactly the degrees of freedom needed to evaluate expectation values of operators in the observable space $\iota(\V)$ that appear in the scope of the functional theory introduced in Definition~\ref{def:scope}.
Practical considerations may, however, motivate the use of reduced variables with redundant degrees of freedom. 
If a variable other than the density defined by the density map is used while the scope is kept fixed, we refer to the resulting formulation as a \emph{reducible} functional theory. 
This formulation involves an extended variable $\eta$ together with a reduction map $R$ such that the irreducible density is recovered as $\rho=R(\eta)$. 
This can be beneficial for having more information about the state stored directly in the variable $\eta$. 
The coupling to the variable external potential is still determined by $\rho$, but additional fixed contributions to the Hamiltonian may become linear functionals of the extended variable. This concept was introduced in the context of RDMFT and symmetry-based variants in Refs.~\cite{Liebert2023-refining, LMS24-jctc, Liebert25-phdthesis}.

An example of this is obtained by using the full 1RDM as the reduced variable while keeping the scope restricted to local external potentials. The irreducible density is then recovered from the 1RDM. Although the observable space $\iota(\V)$ from Definition~\ref{def:scope} remains the space of local external potentials, the kinetic energy can be written exactly as a linear functional of the 1RDM. Thus, this contribution can be separated from the constrained-search functional, leaving the interaction-only universal functional of RDMFT. In this sense, the 1RDM is an extended reduced variable relative to the local-potential scope, even though it is irreducible for the larger scope of all one-body potentials.
Another example arises in spin-resolved formulations. If the scope contains only spin-independent external potentials, the corresponding irreducible density is the spin-summed charge density. One may nevertheless work with the spin-up and spin-down densities, or with spin-resolved blocks of the 1RDM, as extended variables. The spin-summed density is then obtained by a reduction map, while observables related to magnetization become accessible
as linear functionals of the extended variable.

The framework developed in this paper can also be applied to such reducible formulations by keeping track of the extended variable and its reduction to the density associated with the fixed scope. A systematic comparison of the practical advantages of irreducible and reducible formulations is left for future work. We expect that the unified framework developed here will provide useful tools for such investigations.


Taken together, the extensions and generalizations discussed above illustrate the flexibility of the framework developed in this work. More importantly, they reinforce the central message of the paper: many features commonly associated with particular functional theories are in fact consequences of a common underlying structure. By making this structure explicit through the notion of a defining scope, we place a broad class of functional theories within a single mathematical framework.

\begin{acknowledgments}
We would like to thank Chi-Kwong Li for insightful discussion. 
Financial support from the German Research Foundation (Grant SCHI 1476/1-1) (C.-C.~W., M.~P., C.~S.), the Munich Center for Quantum Science and Technology (C.~S.), and the ERC-2021-STG under grant agreement No.~101041487 REGAL (M.~P.) is acknowledged.
\end{acknowledgments}

\appendix
\renewcommand{\theequation}{\thesection.\arabic{equation}}
\renewcommand{\thetheorem}{\thesection.\arabic{theorem}}

\section{Double qubit functional theory}
\label{app:2qubit}

Here, we again consider the functional theory defined on $\mathbb{C}^2\otimes \mathbb{C}^2$ with observables $\{X_1, Y_1, Z_1, Z_2\}$, which is discussed in Section~\ref{sec:geometry_double_qubit}. For the fixed operator of the functional theory we take $H_0 = X_1X_2 + Y_1Y_2$, so that the total Hamiltonian is
\begin{equation}
H(v) = X_1X_2 + Y_1 Y_2
+ v_1 X_1 + v_2 Y_1 + v_3 Z_1 + v_4 Z_4,
\end{equation}
where $v = (v_1, v_2,v_3,v_4) \in \mathbb{R}^4$.
For this simple model, it is possible to perform the pure-state constrained search analytically, which we will do for a density $\rho = (0,0,z,0)$ with $z\in (-1,0) \cup (0,1)$ (see Fig.~\ref{fig:c2c2_diagram}).
\begin{figure}
\centering
\includegraphics[width=.3\columnwidth]{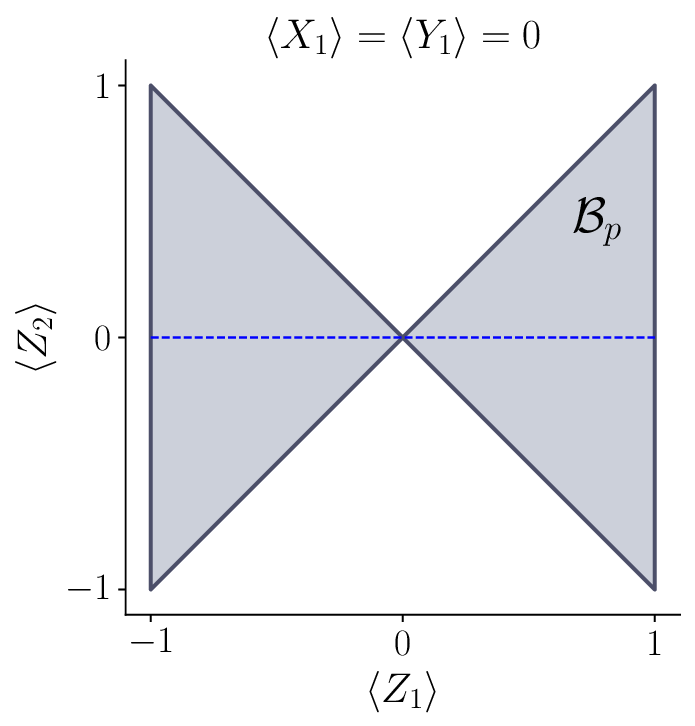}
\caption{The pure-state constrained search is performed along the blue dashed line lying in the pure-state observable range $\Bpure$. In the figure, we only show the slice of $\Bpure$ on which $\braket{X_1}$ and $\braket{Y_1}$ vanish.}
\label{fig:c2c2_diagram}
\end{figure}
For any pure state
$\Psi = (a,b,c,d)$
written in the basis $\{\ket{\uparrow\uparrow},\ket{\uparrow\downarrow},\ket{\downarrow\uparrow},\ket{\downarrow\downarrow}\}$, the density, i.e., the image under the density map $\densmap$, is
\begin{equation}
\densmap(P_\Psi)
= 
\begin{pmatrix}
2\Re(\overline a c + \overline b d)\\
2\Im(\overline a c+ \overline b d)\\
|a|^2 + |b|^2 - |c|^2 - |d|^2\\
|a|^2 - |b|^2 + |c|^2 - |d|^2
\end{pmatrix}.
\end{equation}
Thus, $P_\Psi$ is in the preimage $\densmap^{-1}(0,0,z,0)$ if and only if the complex coefficients $a,b,c,d$ satisfy
\begin{align}
&\overline a c + \overline b d = 0,\\
&|a|^2 - |d|^2 = \frac{z}{2},\\
&|b|^2 - |c|^2 = \frac{z}{2},\\
&|a|^2 + |b|^2 = \frac{1+z}{2},\\
&|c|^2 + |d|^2 = \frac{1-z}{2}.
\end{align}
The last four equations are linear in $|a|^2, |b|^2, |c|^2, |d|^2$. The solution set to these four equation can be written in parameterized form as
\begin{equation}
\begin{pmatrix}
|a|^2\\
|b|^2\\ |c|^2\\ |d|^2
\end{pmatrix}
= \frac{1}{4}\begin{pmatrix}
1+z-t \\
1+z+t \\
1-z+t \\
1-z-t
\end{pmatrix},
\end{equation}
with $t\in \mathbb{R}$ such that all four entries are nonnegative. 
The condition $\overline a c + \overline b d = 0$ then implies $t=0$. Hence, we have for a quantum state $\Psi$ in the preimage that
\begin{equation}
\Psi = \frac{1}{2}
\begin{pmatrix}
\sqrt{1+z} \\ \sqrt{1+z}e^{i\theta} \\ \sqrt{1-z}e^{i\phi} \\ \sqrt{1-z}e^{i\chi}
\end{pmatrix},
\end{equation}
where $\theta, \phi, \chi$ are phases. 
Now, the condition $\overline a c + \overline b d = 0$ is equivalent to $-\theta + \chi = \phi + \pi$, 
so the preimage is
\begin{equation}\label{eq:qubit_psi}
\begin{aligned}
&\densmap^{-1}(0,0,z,0)
= 
\left\{
P_{\Psi(\theta,\phi)}: \Psi(\theta,\phi) = 
\frac{1}{2} 
\begin{pmatrix}
\sqrt{1+z} \\ \sqrt{1+z}e^{i\theta} \\ \sqrt{1-z} e^{i\phi} \\ -\sqrt{1-z} e^{i(\theta+\phi)}
\end{pmatrix},\;
\theta,\phi \in [0, 2\pi]
\right\}.
\end{aligned}
\end{equation}
This is exactly the set of pure states that satisfy the constraint in the definition of the pure-state constrained-search functional (see Definition~\ref{def:const-search-func}).

We will now carry out the constrained minimization.
With $H_0 = X_1X_2 + Y_1Y_2$ and $\Psi(\theta,\phi)\in \densmap^{-1}(0,0,z,0)$ given in Eq.~\eqref{eq:qubit_psi}
we have 
\begin{equation}
\langle\Psi(\theta,\phi),H_0\Psi(\theta,\phi)\rangle = \sqrt{1-z^2}\cos(\theta - \phi).
\end{equation}
It follows that $\Psi(\theta, \phi)$ is a minimizer of the constrained search if and only if $\theta = \phi + \pi$, and the value of the pure-state constrained-search functional at $(0,0,z,0)$ is simply
\begin{equation}
\label{eq:2qubit_fpure}
    \Fpure(0,0,z,0) = -\sqrt{1-z^2}.
\end{equation}
Since $(0,0,z,0)$ is a regular value (recall that we assume $z \in (-1,0)\cup (0,1)$, so the critical values are left out), we can apply Theorem~\ref{th:v-rep}\ref{th:v-rep:item:pure}. This means that for each minimizers $\Psi(\phi+\pi,\phi)$ of the pure-state constrained search at $(0,0,z,0)$ (with $\phi \in [0,2\pi]$ arbitrary), there exists a unique potential $v = v(\phi) \in \mathbb{R}^4$ such that $\Psi(\phi+\pi,\phi)$ is an 
eigenstate of $H(v)$. 
By putting the solutions into the Schrödinger equation $H(v)\Psi=E\Psi$, one finds after some algebra that
\begin{equation}
v(\phi) = \begin{pmatrix}
 \frac{\cos \phi}{z} \\ \frac{\sin\phi}{z} \\ 
 -\sqrt{\frac{1+z}{1-z}} + \frac{1}{2z}\Big(\sqrt{\frac{1+z}{1-z}} - \sqrt{\frac{1-z}{1+z}}\Big)\\
 -\frac{\sqrt{1-z^2}}{z}
\end{pmatrix}.
\end{equation}
We were able to numerically verify that each $\Psi(\phi+\pi,\phi)$ is the \emph{first excited state} of $H(v(\phi))$,
with a finite energy gap to the ground-state energy.
Note that every such $\Psi(\phi+\pi,\phi)$ maps to the \emph{same} $\rho=(0,0,z,0)$, so we have established with this example that pure excited-state $v$-representability can be non-unique.
This also means that $(0,0,z,0)$ is \emph{not} 
pure-state $v$-representable for any $z\in (-1,0)\cup(0,1)$. Hence, it is possible that $\Fens(0,0,z,0) \neq \Fpure(0,0,z,0)$, which is confirmed by the following proposition.

\begin{proposition}
For the two-qubit functional theory, the ensemble constrained-search functional is strictly 
smaller than the pure-state constrained-search functional at $\rho = (0,0,\frac{1}{2}, 0)$.

\begin{proof}
It suffices to find a mixed state $\Gamma$ such that $\densmap(\Gamma) = (0,0,\frac{1}{2},0)$ and 
$\trace(\Gamma H_0) < \Fpure(0,0,\frac{1}{2},0) = -\sqrt{3}/2$ (see Eq.~\eqref{eq:2qubit_fpure}).\\
Let $\Psi_1 = \ket{\uparrow\uparrow}$ and $\Psi_2 = -\sqrt{2/3}\ket{\uparrow\downarrow} + 1/\sqrt{3}\ket{\downarrow\uparrow}$, and take $\Gamma = \frac{1}{4}P_{\Psi_1} + \frac{3}{4}P_{\Psi_2}$. Then
\begin{equation}
\begin{aligned}
&\trace(\Gamma X_1) = \trace{\Gamma Y_1} = 0,\\
&\trace(\Gamma Z_1) = \frac{1}{4}\cdot 1 + \frac{3}{4} \cdot \frac{1}{3} = \frac{1}{2},\\
&\trace(\Gamma Z_2) = \frac{1}{4}\cdot 1 + \frac{3}{4} \cdot \frac{-1}{3} = 0.
\end{aligned}
\end{equation}
So $\densmap(\Gamma) = (0,0,\frac{1}{2},0)$. Furthermore,
\begin{equation}
\trace(\Gamma H_0) = \frac{1}{4}\cdot 0 + \frac{3}{4}\cdot \frac{-4\sqrt{2}}{3} = -\sqrt{2} < -\frac{\sqrt{3}}{2}.
\end{equation}
This shows that the ensemble constrained-search functional is strictly smaller than $\Fpure(0,0,\frac{1}{2},0)$ at this point.
\end{proof}
\end{proposition}

\begin{figure}
\includegraphics[width=.5\columnwidth]{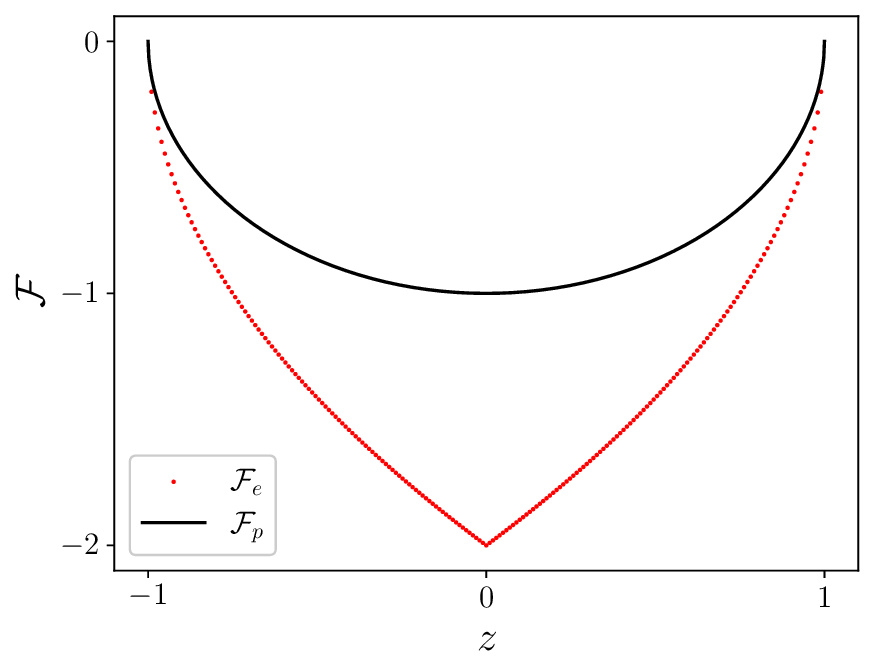}
\caption{Values of the pure-state constrained-search functional $\Fpure$ and the ensemble constrained-search functional $\Fens$ at densities $\rho = (0,0,z,0)$, with $z\in (-1,0)\cup (0,1)$ (see Fig.~\ref{fig:c2c2_diagram}).}
\label{fig:2qubit_fpure_fens}
\end{figure}

Note that the mixed state $\Gamma$ above is actually a minimizer of the ensemble constrained search at $(0,0,\frac{1}{2}, 0)$. In fact, $\Psi_1$ and $\Psi_2$ are degenerate ground states of the Hamiltonian
\begin{equation}
H\left(0, 0, -\sqrt{2}, -\frac{1}{\sqrt{2}}\right)
= X_1X_2 + Y_1Y_2 
- \sqrt{2}Z_1 -\frac{1}{\sqrt{2}}Z_2,
\end{equation}
which is consistent with Lemma~\ref{lem:degen} since $\Fens$ and $\Fpure$ do not agree at $(0,0,\frac{1}{2}, 0)$.

In Fig.~\ref{fig:2qubit_fpure_fens}, we compute the ensemble constrained-search functional $\Fens$ at $\rho = (0,0,z,0)$ for $z\in (-1,1)$ by numerically performing the Legendre--Fenchel transform (Lieb optimization) $\Fens(\rho) = \sup_{v}(E(v) -\langle\rho,v\rangle)$. We then compare it with the pure-state constrained-search functional, which is given by Eq.~\eqref{eq:2qubit_fpure}. Note also that $\Fens$ has a kink at $z=0$, indicating that the density $(0,0,0,0) \in \mathbb{R}^4$ is ensemble $v$-representable by multiple potentials.

\section{Discontinuous ensemble functional example}
\label{app:discont-ens}

Here, we will construct a functional theory whose ensemble functional $\Fens$ is discontinuous on the boundary of its domain. 
Take $\mathcal{H} = \mathbb{C}^3 = \mathbb{C} \oplus \mathbb{C}^2$, $\V=\mathbb{R}^2$, and let
\begin{equation}
\iota(v)=
v_1 (1\oplus X) + v_2 (0\oplus Y)
=\begin{pmatrix}
    v_1&0&0\\
    0&0 & v_1-iv_2\\
    0&v_1+iv_2 & 0
\end{pmatrix}
\end{equation}
where $1\oplus X$ is a block-diagonal $3\times 3$ matrix. It is easy to check that the ensemble observable range in this case is given by the closed unit disk, $\Bens = \{(x,y) \in \mathbb{R}^2 : x^2 + y^2 \le 1\}$.
For the fixed part $H_0$ of the functional theory we take
\begin{equation}
H_0 = -1 \oplus 0 = \begin{pmatrix}-1 \\ & 0 \\ && 0\end{pmatrix}.
\end{equation}
We now claim that the ensemble functional $\Fens$ is discontinuous at $(1,0)\in \Bens$. To this end, we will show that $\Fens(\cos\theta, \sin\theta) = 0$ for $\theta\neq 0$, and $\Fens(1, 0) =-1$.

For this, let $\ket{0}, \ket{1}, \ket{2}$ denote the standard orthonormal basis on $\mathbb{C}^3$. The ensemble state $P_{\ket{0}}$ satisfies $\densmap(P_{\ket{0}}) = (1,0)$, so $\Fens(1,0) \le \braket{0|H_0|0} = -1$. On the other hand, $\Fens \ge -1$ everywhere because $-1$ is the least eigenvalue of $H_0$. Hence, we must have $\Fens(1, 0) = -1$.

Let $\Gamma$ be any ensemble state for which $\densmap(\Gamma) = (\cos\theta, \sin\theta)$ for any $\theta\neq 0$. Write
\begin{equation}
\Gamma = \begin{pmatrix}
\lambda & *\\
* & (1-\lambda)\Gamma'
\end{pmatrix},
\end{equation}
where $\lambda \in [0,1]$ and $\Gamma'$ is a $2\times 2$ matrix with unit trace. We have
\begin{equation}
(\cos\theta, \sin\theta) = \lambda (1,0) + (1-\lambda) (\trace(\Gamma'X), \trace(\Gamma'Y)).
\end{equation}
Note that $(\trace(\Gamma' X), \trace(\Gamma' Y))$ lies within the closed unit disk. Since $(\cos\theta, \sin\theta)$ is an extreme point of the closed unit disk, and $(\cos\theta, \sin\theta)\neq (1,0)$, we conclude $\lambda =0$, which implies $\trace(\Gamma H_0)= 0$. Since this is true for any $\Gamma$ such that $\mu(\Gamma) = (\cos\theta, \sin\theta)$, we have 
$\Fens(\cos\theta, \sin\theta) = 0$ for all angles $\theta\neq 0$.

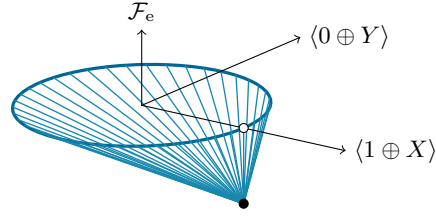
\begin{figure}
    \centering
    \begin{tikzpicture}[line cap=round, line join=round, 
        x={(0.9cm,-0.2cm)}, y={(0.7cm,0.3cm)}, z={(0,1cm)},
        draw=MidnightBlue]
    \xdef\N{40} 
    \xdef\D{1.5}
    \begin{scope}[canvas is xy plane at z=0]
        \draw[line width=1.2pt] (0,0) ellipse (\D cm and \D cm);
        \foreach \i in {1,...,\N} {
            \pgfmathsetmacro{\ang}{((\i-1)*360/\N)};
            \coordinate (top\i) at (\ang:\D cm and \D cm);
        }
    \end{scope}
    \foreach \i in {1,...,\N} {
        \draw[color=MidnightBlue!70,line width=.5pt] (top\i) -- (\D,0,-1);
    }
    \draw[->, black, thin] (0,0,0) -- (3,0,0) node[right] {$\braket{1\oplus X}$};
    \draw[->, black, thin] (0,0,0) -- (0,3,0) node[right] {$\braket{0 \oplus Y}$};
    \draw[->, black, thin] (0,0,0) -- (0,0,1) node[above] {$\Fens$};
    \filldraw[black, fill=white] (\D,0,0) circle (1.7pt);
    \filldraw[black, fill=black] (\D,0,-1) circle (1.7pt);
    \end{tikzpicture}
    \caption{
        Example of an ensemble constrained-search functional $\Fens$ defined on a disc that is discontinuous on the boundary.
    }
    \label{fig:Fens-discontinuous-ex}
\end{figure}

Through the direct-sum construction, the whole ensemble constrained-search functional is simply given by the convex hull of its boundary values. The functional is illustrated in Figure~\ref{fig:Fens-discontinuous-ex}.

\section{Discontinuous pure functional example}
\label{app:discont-pure}

Here, we will construct a functional theory whose pure functional $\Fpure$ is discontinuous on the interior of its domain. Take $\H=\C^3$, $\V=\R^3$, and let
\begin{equation}
    \iota(v)=\frac{v_1}{2}\begin{pmatrix} 0&1&0 \\ 1&0&1 \\ 0&1&0 \end{pmatrix} + 
    \frac{v_2}{2}\begin{pmatrix} 0&-\rmi&0 \\ \rmi&0&\rmi \\ 0&-\rmi&0 \end{pmatrix} + 
    v_3\begin{pmatrix} 1 \\ &0 \\ &&-1 \end{pmatrix}.
\end{equation}
For the fixed part $H_0$ of the functional theory we take
\begin{equation}
H_0 = \begin{pmatrix} 0 \\ & -9000 \\ && 0\end{pmatrix}.
\end{equation}
Let $\ket{1}, \ket{0}, \ket{-1}$ be the standard orthonormal basis on $\mathbb{C}^3$, with respect to which we write any pure state as
\begin{equation}
    \Psi = a\ket{1}+ b\ket{0}+ c\ket{-1},\quad a,b,c\in\C.
\end{equation}
The density map can then be evaluated as 
\begin{equation}
\mu(P_\Psi)=(\Re(b(\overline{a}+\overline{c})),\Im(b(\overline{a}+\overline{c})),|a|^2-|c|^2).
\end{equation}
We will consider the pure functional $\Fpure$ at density $\rho=(0,0,z)$ for different values of $z$, where it was checked that $\rho=(0,0,0)$ indeed lies in the interior of the observable range $\Bpure$. Such densities imply $b(\overline{a}+\overline{c})=0$ and $|a|^2-|c|^2=0$ for the coefficients. We distinguish two cases.

For $z\neq 0$ we can only have $b=0$, which leads to $|a|^2=(1+z)/2$ and $|c|^2=(1-z)/2$. If it would hold $b\neq 0$ then from $b(\overline{a}+\overline{c})=0$ it follows $a=-c$ which is incompatible with $z=0$. Since $H_0$ only gives a non-zero contribution if $b\neq 0$ we get $\Fpure(0,0,z)=0$ for all $z\in[-1,1]$.

For $z=0$ we must have $|a|=|c|$ and can allow $b\neq 0$ which lowers the expectation value $\langle\Psi,H_0\Psi\rangle$ that enters the pure-state constrained-search functional. It is clearly maximized at $|b|=1$ and $a=c=0$, so that we get $\Fpure(0,0,0)=-9000$, which makes the functional discontinuous.

\section{Lower semicontinuity, but not continuity, of a pure-state RDMFT Hubbard dimer functional \label{app:discont-real}}

In this section, we show that for the two-site Hubbard dimer at half-filling in the singlet subspace the pure-state constrained-search functional is lower semi-continuous but not continuous on $\Bpure^\R$, as mentioned without explicit proof in Ref.~\cite{Liebert2023-refining}, when one only considers \emph{real} coefficients.
We work with the real $2\times 2$ spin block of the 1RDM, normalized so that $\trace\gamma=1$. In the orthonormal basis
\begin{align}
\Phi_1 &= c_{1\uparrow}^* c_{1\downarrow}^* |0\rangle, \\
\Phi_2 &= c_{2\uparrow}^* c_{2\downarrow}^* |0\rangle, \\
\Phi_3 &= \frac{1}{\sqrt{2}}
\bigl(c_{1\uparrow}^* c_{2\downarrow}^* -c_{1\downarrow}^* c_{2\uparrow}^*\bigr)|0\rangle,
\end{align}
a normalized real singlet state has the form
\begin{equation}\label{eq:app-real-state}
\Psi = a\Phi_1+b\Phi_2+c\Phi_3,\quad a,b,c\in\mathbb R, \quad a^2+b^2+c^2=1.
\end{equation}
For such a state, the corresponding real 1RDM is parametrized by
\begin{equation}\label{eq:app-gamma-constraints}
\gamma_{11}=a^2+\frac{c^2}{2}, \quad \gamma_{12}=\frac{c(a+b)}{\sqrt{2}}.
\end{equation}
Therefore, the domain of the constrained-search functional that only varies over states of the form Eq.~\eqref{eq:app-real-state} is the disk
\begin{equation}\label{eq:app-domain}
\Bpure^\R= \left\{ (\gamma_{11},\gamma_{12})\in\mathbb R^2: \left(\gamma_{11}-\frac12\right)^2+\gamma_{12}^2\le \frac14 \right\}.
\end{equation}
This will make the resulting functional $\Fpure^\R$ depend on the basis choice $\{\Phi_1,\Phi_2,\Phi_3\}$, but note that it in quantum chemistry typically such expansions with real coefficients are considered.

We restrict our attention to the repulsive on-site interaction
\begin{equation}\label{eq:app-W}
H_0=W=U (n_{1\uparrow}n_{1\downarrow}+n_{2\uparrow}n_{2\downarrow}), \quad U>0 ,
\end{equation}
for which
\begin{equation}
\langle\Psi,W\Psi\rangle = U(a^2+b^2) = U(1-c^2).
\label{eq:app-W-exp}
\end{equation}
Then, it can be shown that for $\gamma\neq \gamma_*:=(\frac12,0)$, the constrained search yields the closed formula \cite{Pastor2011a,Cohen2016,Liebert2023-refining}
\begin{equation}\label{eq:app-F}
\begin{aligned}
&\Fpure^\R(\gamma)= U \frac{ \left(\gamma_{11}-\frac12\right)^2 +\frac12\,\gamma_{12}^2 \left[ 1-\sqrt{1-4\left(\gamma_{11}-\frac12\right)^2-4\gamma_{12}^2} \right]} { \left(\gamma_{11}-\frac12\right)^2+\gamma_{12}^2 }.
\end{aligned}
\end{equation}
The point $\gamma_*=(\frac12,0)$ is the only point at which Eq.~\eqref{eq:app-F} is not defined a priori because the denominator vanishes.
Then the constrained search value of the pure-state functional at $\gamma_*=(\frac12,0)$ for $U>0$ is $\mathcal F^\R_{\mathrm p}(\gamma_*)=0$.
In the following we show that $\Fpure^\R$ is continuous on $\Bpure^\R\setminus\{\gamma_*\}$, lower semicontinuous on all of $\Bpure^\R$, and not continuous at $\gamma_*$.

We split the proof of the above statement into four steps, and first show the value of $\mathcal F_{\rm p}$ at the exceptional point $\gamma_*$.
For this, the constraints of Eqs.~\eqref{eq:app-gamma-constraints} and \eqref{eq:app-real-state} become
\begin{equation}\label{eq:app-center-constraints}
a^2+\frac{c^2}{2}=\frac12,\quad  \frac{c(a+b)}{\sqrt{2}}=0, \quad a^2+b^2+c^2=1.
\end{equation}
The first and third relations imply $a^2=b^2=\frac{1-c^2}{2}$. Conversely, for every $c\in[0,1]$, the choice
\begin{equation}
a=\sqrt{\frac{1-c^2}{2}}, \quad b=-\sqrt{\frac{1-c^2}{2}}
\end{equation}
satisfies Eq.~\eqref{eq:app-center-constraints}. Therefore every $c\in[0,1]$ is admissible at $\gamma_*$, and thus $\Fpure^\R=0$ since $U>0$.

Next, we show continuity away from $\gamma_*$. For this, set $x:=\gamma_{11}-\frac12$ and $y:=\gamma_{12}$.
Then on $\Bpure^\R\setminus\{\gamma_*\}$, $\mathcal F_{\rm p}(\gamma)$ in Eq.~\eqref{eq:app-F} reads
\begin{equation}\label{eq:app-xy-form}
\Fpure^\R(x,y) = U \frac{x^2+\frac12 y^2\left(1-\sqrt{1-4x^2-4y^2}\right)}{x^2+y^2}.
\end{equation}
Because $(x,y)\in \Bpure^\R$, one has $x^2+y^2\le \frac14$, and hence $1-4x^2-4y^2\geq 0$.
The numerator of Eq.~\eqref{eq:app-xy-form} is continuous on $\Bpure^\R$, and the denominator $x^2+y^2$ is strictly positive on $\Bpure^\R\setminus\{\gamma_*\}$. Hence
$\Fpure^\R(\gamma)$ is continuous on $\Bpure^\R\setminus\{\gamma_*\}$. 

To show lower semicontinuity at $\gamma_*$, we observe that it suffices to verify the sequential criterion.
Let $\gamma_n\to\gamma_*$ in $\Bpure^\R$. If $\gamma_n=\gamma_*$ for infinitely many $n$, then
\begin{equation}
\liminf_{n\to\infty}\Fpure^\R(\gamma_n)\geq 0 = \Fpure^\R(\gamma_*),
\end{equation}
and there is nothing to prove. We therefore assume that
$\gamma_n\neq\gamma_*$ for all $n$ and consider
\begin{equation}
x_n:=\gamma_{11}^{(n)}-\frac12, \quad y_n:=\gamma_{12}^{(n)}.
\end{equation}
Then, Eq.~\eqref{eq:app-xy-form} yields
\begin{equation}
\Fpure^\R(\gamma_n) = U \frac{x_n^2+\frac12 y_n^2\left(1-\sqrt{1-4x_n^2-4y_n^2}\right)}{x_n^2+y_n^2}.
\end{equation}
In addition, $x_n^2\geq 0$, $y_n^2\geq 0$, and $0\leq \sqrt{1-4x_n^2-4y_n^2}\leq 1$, and therefore
\begin{equation}
1-\sqrt{1-4x_n^2-4y_n^2}\geq 0.
\end{equation}
Since $U>0$, every term in the numerator is nonnegative, hence
\begin{equation}
\Fpure^\R(\gamma_n)\geq 0 \quad\text{for all } n.
\end{equation}
Thus, it follows that
\begin{equation}
\liminf_{n\to\infty}\Fpure^\R(\gamma_n)\geq 0 =  \Fpure^\R(\gamma_*)
\end{equation}
and $\Fpure^\R$ is lower semicontinuous at $\gamma_*$. Combined with Step~2, this proves that $\Fpure^\R$ is lower semicontinuous on all of $\Bpure^\R$.

To finally show failure of continuity at $\gamma_*$,  we consider the two sequences
\begin{equation}
\gamma_n^{(1)}:=\left(\frac12+\frac1n,0\right), \quad \gamma_n^{(2)}:=\left(\frac12,\frac1n\right), \quad n\geq 2.
\end{equation}
Both belong to $\Bpure^\R$ and converge to $\gamma_*$. For the first sequence,
\begin{equation}
\lim_{n\to\infty} \Fpure^\R(\gamma_n^{(1)}) =   U.
\end{equation}
Along the second sequence,
\begin{equation}
 \Fpure^\R(\gamma_n^{(2)}) = \frac{U}{2}\left(1-\sqrt{1-\frac{4}{n^2}}\right).
\end{equation}
and, thus, 
\begin{equation}
\lim_{n\to\infty} \Fpure^\R(\gamma_n^{(2)})=0.
\end{equation}
Since $U>0$, these two limits are different. Hence the limit of $ \Fpure^\R(\gamma)$ as $\gamma\to\gamma_*$ does not exist, and $\Fpure^\R$ is not continuous at $\gamma_*$.

The above proof is specific to the repulsive case $U>0$. For $U<0$, the same constrained search argument at $\gamma_*=(\frac12,0)$ gives
\begin{equation}
\Fpure^\R(\gamma_*)=\min_{0\le c\le 1}U(1-c^2) = U,
\end{equation}
and the sign-dependent closed formula tends to the same value $U$ as $\gamma\to\gamma_*$. Hence, for attractive interaction, the functional is continuous at the center of the disk $\Bpure^\R$.

\section{Moment maps}
\label{app:symplectic}

Here, we will review some elements of symplectic geometry that are relevant for functional theories as discussed in Section~\ref{sec:Lie}.

\subsection{Group actions}

Let $G$ be a Lie group acting on $M$ from the left. That is, there is a smooth map
\begin{equation}
\begin{aligned}
G \times M &\rightarrow M \\
(g, p) &\mapsto g\cdot p
\end{aligned}
\end{equation}
that satisfies $g\cdot (h\cdot p) = (g\cdot h)\cdot p$.
For a vector $S\in \mathfrak{g}$, the \emph{fundamental vector field} $S^* \in\mathfrak{X}(M)$ is defined by
\begin{equation}
S^*_p = \frac{\rmd }{\rmd t}\exp (tS) \cdot p\Big|_{t=0}.
\end{equation}
The map $\mathfrak{g} \rightarrow \mathfrak{X}(M)$, $S\mapsto S^*$ is a Lie algebra antihomomorphism. That is,
\begin{equation}
 [S, T]^* = -[S^*, T^*].
\end{equation}

\subsection{Symplectic geometry and moment maps}
Let $(M, \omega)$ be a symplectic manifold. That is, $M$ is a mani\-fold and $\omega \in \Omega^2(M)$ is a nondegenerate two-form satisfying $\rmd\omega = 0$. 
For a point $p\in M$, let $\flat:T_pM \rightarrow T_p^*M$ be the isomorphism $v\mapsto \omega(v, \cdot)$, and let $\sharp: T_p^*M \rightarrow T_pM$ be its inverse.
The \emph{symplectic gradient} of a smooth function $f\in C^\infty(M)$ is
\begin{equation}
\sgrad  f:= (\rmd f)^\sharp.
\end{equation}
For two smooth functions $f,g \in C^\infty(M)$, the \emph{Poisson bracket} is defined as
\begin{equation}
\{f, g\} :=  \omega(\sgrad f, \sgrad g).
\end{equation}
This makes $(C^\infty(M), \{\cdot,\cdot\})$ into a Lie algebra.
The map $\sgrad : C^\infty(M) \rightarrow \mathfrak{X}(M)$ is a Lie algebra antihomomorphism and it holds
\begin{equation}
\sgrad \{f,g\} = -[\sgrad f, \sgrad g].
\end{equation}

Let $G$ be a Lie group acting on a symplectic manifold $(M, \omega)$ from the left. Let 
$\phi: \mathfrak{g}\rightarrow \mathfrak{X}(M)$ denote the antihomomorphism $S\mapsto S^*$.
\begin{definition}
A \emph{comoment map} for the $G$-action is a Lie algebra homomorphism $\Phi:\mathfrak{g}\rightarrow C^\infty(M)$ such that $\sgrad \circ\,\Phi = \phi$.
\end{definition}

There is a one-to-one correspondence between linear maps $\Phi: \mathfrak{\mathfrak{g}}\rightarrow C^\infty(M)$ and smooth maps $\mu: M\rightarrow \mathfrak{g}^*$ given by
\begin{equation}
\braket{\mu(p), X} = \Phi(X)(p).
\end{equation}
If $\Phi$ is a comoment map, the corresponding $\mu$ is called a \emph{moment map} for the $G$-action.

\begin{example}
Let $M = T^* \mathbb{R} = \mathbb{R}^2$ be the cotangent bundle of $\mathbb{R}$. 
Equip $M$ with the symplectic form $\omega = \rmd x \wedge \rmd p$, where $(x,p)$ are the Euclidean coordinates on $\mathbb{R}^2$.
Let $\mathbb{R}$ act on $\mathbb{R}^2$ by translating the first coordinate, i.e., $a\cdot (x,p) = (x+a, p)$. This action admits a moment(um) map
$(x,p) \mapsto p$.
\end{example}

\subsection{Properties of moment maps}

A basic but important property of a moment map is $G$-equivariance.
Recall that the \textit{coadjoint action} $\mathrm{Ad}^* : G\rightarrow \mathrm{GL}(\mathfrak{g}^*)$
is defined to be the dual of the adjoint representation.

\begin{lemma}\label{lem:moment-map-equivariance}
Assume $M$ and $G$ are connected. A moment map $\mu: M\rightarrow \mathfrak{g}^*$ for the $G$-action is $G$-equivariant with respect to the coadjoint action. That is, for all $g\in G$ and $p\in M$, it holds that
\begin{equation}
\mu(g\cdot p) = \Ad^*_g \mu(p).
\end{equation}
\end{lemma}

The following lemma is convenient for constructing new moment maps 
out of old ones:

\begin{lemma}
\label{lem:mom_map_composition}
Let $b: H\rightarrow G$ be a Lie group homomorphism. Let $\mu: M\rightarrow \mathfrak{g}^*$ be a moment map for the $G$-action. Then $b^*\circ \mu: M \rightarrow \mathfrak{h}^*$ is a moment map for the $H$-action.
\end{lemma}

Finally, we state the convexity theorems on the image of a moment map.
Let $G$ be a compact connected Lie group acting on a compact connected symplectic manifold $M$ with moment map $\mu: M\rightarrow \mathfrak{g}^*$. Then:

\begin{theorem}[Atiyah--Guillemin--Sternberg \cite{Atiyah82,GS82}]
\label{thm:atiyah}
Suppose $G$ is abelian (i.e., a torus). Then 
\begin{equation}
\mu(M) = \conv\{\mu(p)\mid p \text{ is } G\text{-fixed point}\}.
\end{equation}
\end{theorem}

\begin{theorem}[Kirwan \cite{Kirwan84}]
\label{thm:kirwan}
Let $\mathfrak{t}\subseteq \mathfrak{g}$ be a Cartan subalgebra. Choose an $\Ad$-invariant inner product $\braket{\cdot,\cdot}$ on $\mathfrak{g}$ so that $\mathfrak{t}^*$ can be identified with a subspace of $\mathfrak{g}^*$. Let $\mathfrak{t}^*_+$ be a positive Weyl chamber. Then
\begin{equation}
\mu(M) \cap \mathfrak{t}^*_+
\end{equation}
is a convex polytope.
\end{theorem}

\bibliography{refs}

\end{document}